%% file: manuscript.tex
\pdfoutput=1

\documentclass[11pt]{article}

\usepackage[margin=1in]{geometry}
\usepackage[T1]{fontenc}
\usepackage[utf8]{inputenc}
\usepackage{lmodern}
\usepackage{microtype}
\usepackage{setspace}
\usepackage{anyfontsize}
\usepackage{siunitx}
\usepackage{eurosym}
\usepackage{amsmath,amssymb,amsfonts,amsthm,mathtools,bm}

\usepackage{graphicx}
\usepackage{booktabs}
\usepackage{array}
\usepackage{multirow}
\usepackage{threeparttable}
\usepackage{caption}
\usepackage{subcaption}
\usepackage{float}
\usepackage{algorithm}
\usepackage{algpseudocode}
\usepackage{tikz}
\usepackage{rotating}
\usepackage{fancyvrb}
\usepackage{placeins}

\usepackage[round,authoryear]{natbib}
\usepackage{enumitem}
\setlist[itemize]{topsep=3pt,itemsep=2pt,parsep=0pt}
\setlist[enumerate]{topsep=3pt,itemsep=2pt,parsep=0pt}

\usepackage{xcolor}
\definecolor{linkblue}{HTML}{1A4C8B}
\usepackage[
    colorlinks=true,
    linkcolor=linkblue,
    citecolor=linkblue,
    urlcolor=linkblue
]{hyperref}
\usepackage[nameinlink,noabbrev]{cleveref}

\theoremstyle{plain}

\newtheorem{proposition}{Proposition}
\newtheorem{lemma}{Lemma}
\newtheorem{corollary}{Corollary}

\theoremstyle{definition}

\theoremstyle{remark}

\usepackage{etoolbox}
\newif\ifinappendix
\inappendixfalse
\pretocmd{\section}{%
  \ifinappendix
    \setcounter{equation}{0}%
    \setcounter{figure}{0}%
    \setcounter{table}{0}%
    \setcounter{lemma}{0}%
  \fi
}{}{}

\newcommand{\startappendices}{%
  \clearpage
  \appendix
  \inappendixtrue
  \section*{Appendices}
  \addcontentsline{toc}{section}{Appendices}
  \renewcommand{\theequation}{\thesection.\arabic{equation}}
  \renewcommand{\thefigure}{\thesection.\arabic{figure}}
  \renewcommand{\thetable}{\thesection.\arabic{table}}
}

\title{Competing firms, competing regulators: The strategic cost of fragmented climate policy}

\author{
  Nicole Adler\thanks{The Hebrew University of Jerusalem, Israel. Email: \href{mailto:nicole.adler@mail.huji.ac.il}{nicole.adler@mail.huji.ac.il}.}
  \and
  Gianmarco Andreana\thanks{University of Bergamo, Italy; The Hebrew University of Jerusalem, Israel. Email: \href{mailto:gianmarco.andrena@unibg.it}{gianmarco.andreana@unibg.it}.}
  \and
  Gerben de Jong\thanks{SEO Amsterdam Economics, the Netherlands; The Vrije Universiteit Amsterdam, the Netherlands. Email: \href{mailto:g.dejong@seo.nl}{g.dejong@seo.nl}.}
}

\date{}

\begin{document}

\maketitle

\begin{abstract}
  Climate policy in global network industries is implemented across fragmented jurisdictions, yet firms respond through integrated operational networks. We develop a two-stage game-theoretic framework to analyze how firm-level responses interact with alternative governance structures. Regulators first choose emissions charges. Firms subsequently compete through pricing, service capacity and capital deployment decisions. The analytical results demonstrate that uniform global regulation maximizes welfare in symmetric markets. However, in sufficiently asymmetric markets, a uniform global charge is dominated by decentralized regimes. Multiple regulatory instruments better accommodate region-specific market externalities. We apply this framework to a calibrated case study of North American, Western European and transatlantic aviation markets. The numerical results establish that a globally coordinated regulator setting region-specific charges achieves the highest aggregate welfare. These aggregate gains nonetheless mask substantial distributional disparities across jurisdictions. Effective climate governance in network industries therefore requires more than determining an efficient emissions charge. Policy instruments ought to accommodate regional heterogeneity and transfer mechanisms will be necessary to ensure efficient, politically stable cooperation.
\end{abstract}

\vspace{0.75em}

\section{Introduction}\label{sec:introduction}

Firms in global network industries make strategic decisions across multiple jurisdictions while facing regulation imposed at regional, national, supranational, and sector-specific levels. Operational networks are integrated, but the regulatory environment is fragmented. A charge, standard, or quota that is local from the regulator's perspective propagates through pricing, capacity, and capital allocation decisions across the network, while a firm's response in one market depends on the regulatory exposure of the rest of its operations. This creates a central problem for both firm strategy and regulatory governance: how should firms organize operations under fragmented and overlapping regulation, and how do firm-level responses reshape the regulatory architectures they face?

The question matters because many of the central policy challenges of the twenty-first century require coordinated action across jurisdictions, while at the same time the effectiveness of such policies depends critically on how firms reorganize operations under heterogeneous regulatory incentives. Climate policy provides a particularly important example, as firms adjust pricing, capacity, and investment decisions across markets in ways that shape the environmental effectiveness and welfare consequences of emissions regulation \citep{nordhaus2019climate, Perino2019, economides1996economics}. Multi-market firms may absorb environmental charges, pass them through to consumers, restrict service capacity, accelerate capital asset renewal, or shift activity toward less regulated jurisdictions. These operational adjustments determine the realized distribution of emissions, consumer surplus, producer surplus and welfare across regions. Sustainable management in network industries therefore requires more than setting the right carbon price. It requires understanding how a firm's strategic decisions interact with the fragmented regulatory environment in which they are made, and how those decisions feed back into the regulatory architecture itself.

Operations management research has examined how environmental regulation shapes firm behavior along several related dimensions. One stream studies carbon pricing and regulatory instruments, showing how emissions taxes, cap-and-trade systems, and related policies affect production, capacity portfolios, technology adoption, clean-technology investment, and operational flexibility under demand and policy uncertainty \citep{drake2016technology, fan2023price}. A second stream links environmental outcomes to supply chain and network structure, showing that facility location, transportation patterns, market coverage, and product offerings jointly determine both costs and emissions, especially when firms operate across markets with heterogeneous regulatory standards \citep{cachon2014retail, han2022curbing}. A third stream examines coordination and responsibility allocation in supply chains, studying how shared, indirect, or upstream emissions should be attributed across firms and how such allocation rules affect abatement incentives, cooperation, and double marginalization \citep{gopalakrishnan2021incentives, gopalakrishnan2021consistent}. Together, this literature shows that environmental regulation in supply chains is not simply a matter of setting a carbon price, but of understanding how policy instruments interact with firms' operational networks, investment incentives, and inter-firm dependencies. We build on this literature by studying a networked setting in which firms respond strategically to fragmented emissions regulation through pricing, capacity, and investment decisions, and in which these responses shape the performance of alternative regulatory schemes.

Aviation provides a particularly transparent setting in which to study this problem. Airlines connect jurisdictions through hub-and-spoke networks and simultaneously face emissions trading systems at national and supranational levels \citep{schmalensee2017lessons, narassimhan2018carbon}, the Carbon Offsetting and Reduction Scheme for International Aviation on international routes \citep{scheelhaase2018eu, jiang2025role, ovaere2025strategic}, ticket taxes at the country level \citep{bernardo2024flight}, and airport-level environmental charges \citep{fageda2025environmental}. The same structural features appear in international shipping, electricity transmission, multi-modal freight, logistics networks and global supply chains, where firms operate over connected networks, demand and congestion are spatially heterogeneous, and regulatory interventions in one part of the network alter outcomes elsewhere. Electricity markets provide a close analogue, since carbon regulation in congested power networks interacts with demand elasticity, transmission constraints, market structure and strategic firm behavior \citep{limpaitoon2011impact}. The mechanisms we identify in aviation therefore generalize to any network industry in which firms transmit regulatory signals across jurisdictions through linked operational decisions.

The theoretical case for carbon pricing is well established \citep{Marron2014, Stiglitz2019, timilsina2022carbon}. Three features of global network industries complicate how firms respond to it. First, market characteristics vary across space. Regional differences in demand elasticity, market power, network externalities and operating costs imply that the same charge produces different firm responses in different markets, and firms may exploit these differences by passing costs through to less elastic demand segments or by reallocating capacity toward markets where regulation is less binding. Second, operational networks transmit local policies across jurisdictions. Environmental charges in one region can induce carbon leakage when firms shift activity toward less regulated markets \citep{bohringer2012role, carbone2013linking}, and overlapping instruments imposed by multiple regulators may either amplify or offset each other depending on how firms restructure operations \citep{Baylis2013, Nordhaus2015, perino2025overlapping}. Third, regulators act strategically. Regional authorities choose policy instruments to maximize local welfare while recognizing that multi-market firms can reallocate activity across jurisdiction. Regulatory competition therefore unfolds through the operational responses of network firms, potentially pushing equilibrium outcomes away from globally efficient outcomes \citep{oates1988economic, kennedy1994equilibrium}.

These dynamics generate ambiguous effects on both emissions and firm strategy. Higher effective carbon prices should encourage firms to deploy cleaner capital assets \citep{deJong2022fleet} and moderate service capacity \citep{fageda2022pricing}. Market power at hub facilities \citep{brueckner2002airport, pels2004economics}, positive frequency externalities \citep{mohring1972optimization}, cost heterogeneity, and strategic regulatory interaction may push firm responses in the opposite direction. A firm that reduces emissions in one jurisdiction may simply shift them elsewhere. A regulatory regime that improves global welfare may leave some firms or regions worse off and therefore fail to be politically stable. Resolving these ambiguities requires a structural framework that integrates firm-level operational responses with the incentives of the regulators who set the environment in which those decisions are made.

This paper develops a two-stage game-theoretic model of firm operations under fragmented sustainability governance in a multi-market network industry. In the first stage, regulators at regional and global levels choose emission charges to maximize jurisdiction-specific or global welfare. In the second stage, firms observe these charges and compete through pricing, capacity and capital deployment decisions. The firm decision problem is the center of the framework. The regulatory stage matters because it defines the environment firms face, and because firm behavior in the second stage shapes regulatory incentives in the first. We characterize equilibrium outcomes under five governance regimes, namely baseline, decentralized regional regulation, uniform global regulation, differentiated global regulation and overlapping regulation. These regimes jointly span two distinct dimensions of the regulatory environment firms face, coordination and instrument flexibility. Coordination refers to whether one authority internalizes emissions and surplus across regions. Instrument flexibility refers to whether the authority sets one common charge or region-specific charges. Comparing these regimes isolates how firm strategy varies with the regulatory architecture and how that architecture in turn responds to firm behavior.

The paper makes three primary contributions to the literature on sustainable operations management in regulated network industries. First, we identify the operational channels through which multi-market firms transmit environmental regulation across a network. Firms respond to emissions charges through three linked margins: service frequency, fleet composition, and fare pass-through. These margins vary systematically with business model, network geography, passenger heterogeneity, and regulatory exposure. In the calibrated case study, low-cost carriers operating short-haul point-to-point networks adjust primarily through capacity contraction, whereas hub-and-spoke legacy carriers combine smaller frequency reductions with larger fare increases concentrated in less elastic demand segments. These results show that environmental regulation in network industries acts as a system-wide operational shock rather than a route-level cost shock, because shared fleet constraints and network connectivity propagate local charges across markets. To our knowledge, this is the first paper to derive these firm-response channels structurally within an endogenous, multi-regulator setting and to show that the welfare ranking across governance regimes depends on which operational channel dominates.

Second, we characterize how market structure changes the regulatory environment faced by firms. The efficient emissions charge generally differs from the Pigouvian benchmark because it reflects not only marginal environmental damages but also pre-existing distortions from market power, frequency choices, and demand heterogeneity. As a result, the optimal charge may fall below the social cost of carbon and, in some cases, become a subsidy. We show that independent regional regulators impose excessive aggregate charges in symmetric markets, and that overlapping regional and global regulation further amplifies this inefficiency. However, when regions differ in costs, demand, or emissions intensity, the welfare ranking across governance regimes can reverse. A uniform global charge may be too coarse to accommodate regional heterogeneity, so decentralized or overlapping regulation can dominate uniform global regulation despite regulatory competition. We identify this reversal analytically in a symmetric duopoly extension and confirm it numerically in an asymmetric two-region network. This result extends the environmental-federalism literature \citep{williams2012growing, banzhaf2012fiscal, coria2021interjurisdictional} to network industries in which firms transmit regulatory signals across markets rather than operating as independent facilities.

Third, we quantify these mechanisms in a calibrated aviation case study and show that the relevant regulatory benchmark for firms is not necessarily the welfare-maximizing regime, but the regime that can be sustained politically. The extended model is applied to North American, Western European, and transatlantic aviation markets with heterogeneous airlines, nested logit demand, endogenous fleet choices, and hub-and-spoke routing. Differentiated global regulation achieves the highest aggregate welfare, but it leaves North America worse off than under decentralized regional regulation, creating an incentive to defect. We compute the side-payment range required to sustain cooperation, approximately \euro 7 to \euro 10 billion annually, and show that the direction of transfers can reverse under lower carbon damages, when the global regulator optimally subsidizes the industry to mitigate market-power distortions. The case study therefore shows that fragmented regulation can persist even when coordinated regulation is welfare superior. For firms, this implies that capacity, fleet renewal, and pricing decisions should be evaluated under fragmented regulatory exposure rather than under the assumption that global coordination is the natural long-run outcome.

The remainder of the paper is organized as follows. Section~\ref{sec:literature} reviews the related literature. Section~\ref{sec:analytical_model} presents the analytical model, derives the regulators' best responses and the resulting equilibrium and welfare rankings across regulatory regimes. We also analyze asymmetric markets in which regions differ in market structure and flight distance. Section~\ref{sec:extended_model} develops an extended, choice-based numerical model that embeds nested logit demand, hub-and-spoke routing and endogenous fleet choice within the same two-stage game. Section~\ref{sec:case_study} applies this model to the North American, Western European and transatlantic aviation markets, comparing airline responses and welfare outcomes across governance regimes and quantifying the side payments required to sustain cooperation. Section~\ref{sec:conclusions} concludes.

\section{Related Literature}
\label{sec:literature}
This research connects three bodies of literature. First, how do firms in network industries adjust their operations in response to environmental regulation? Second, how do multiple regulators at regional and global levels interact when setting environmental charges on the same multi-market firms? Third, how do pre-existing market distortions, including market power, frequency externalities and carbon leakage, alter the effectiveness of these charges? Comprehensive reviews of the individual domains are available elsewhere \citep{de2012transport, zhang2012airports, stavins2022relative}.

\subsection{Firm strategic responses to environmental regulation}

The operations management literature establishes that firm-level responses to emissions regulation are strategic rather than passive. \citet{drake2016technology} and \citet{krass2013environmental} show that the form of regulation affects both the direction and magnitude of firms' technology and capacity choices. \citet{cachon2014retail} and \citet{park2015supply} demonstrate that network configuration and supply-chain design mediate the environmental and welfare consequences of carbon pricing, so that a carbon price alone is an ineffective instrument when network structure shapes firm and consumer behavior. We extend these insights to a multi-jurisdictional setting in which overlapping regulators and multi-market firm competition jointly determine equilibrium outcomes.

A parallel game-theoretic literature models firm competition in network industries through prices, frequencies and network design \citep[e.g.,][]{Hansen1990, Hong1992, Dobson1993, Hendricks1999, Adler2005, Vaze2012, Adler2014}, reviewed in \citet{adler2021review}; see also \citet{doyme2019simulating}. This literature has developed rich models of oligopolistic competition on networks but has rarely incorporated environmental policy. The strand closest to our second-stage model studies how landing fees and capacity choices shape aircraft size, seat availability and service frequency in duopoly markets \citep{wei2005impact, wei2006impact, wei2007airlines}. Our second-stage model also relates to choice-based airline planning that integrates passenger choice with timetabling, schedule design, fleet assignment and network planning \citep{wei2020airline, yan2022choice, birolini2021airline}.

Studies of firm responses to environmental regulation in aviation typically take the regulatory regime as exogenous. \citet{Brueckner2010} show that a well-designed emission charge in a duopoly improves welfare through fleet upgrades and more efficient network choices, while \citet{Yuen2011} show that unilateral greenhouse-gas regulation can induce inefficient networks by distorting competition between domestic and foreign carriers. Further work examines how emissions trading, allowance allocation and offset schemes affect entry, competition and collusion \citep{barbot2014trade, Sheu2014, Zheng2019}, and how the EU ETS and CORSIA compare as instruments \citep{anger2010including, Vespermann2011, scheelhaase2018eu}. \citet{chen2024impacts} examine sustainable aviation fuel mandates, and \citet{adler2026estimating} estimate the nested logit demand function used in our case study. Simulation and econometric studies extend this work to ticket prices, demand, network configurations and fleet decisions \citep{Sgouridis2011, Kahn2016, Oesingmann2022}. \citet{chen2026carbon} estimate a structural model of European airline competition with endogenous entry, frequencies and pricing under the EU ETS phase-out of free allowances, finding asymmetric impacts across carrier types and segments but holding the regulatory regime fixed.

The closest antecedent is \citet{ovaere2025strategic}, who model a multi-stage game with governments endogenously setting fuel taxes and efficiency standards followed by competition between aircraft manufacturers. We share their structural approach of endogenizing regulatory decisions within a game-theoretic model of industry competition, but study downstream carrier operations rather than upstream technology markets. The operational responses of multi-market carriers, including cost pass-through across heterogeneous demand segments and fleet redeployment across jurisdictions, are the primary channel through which overlapping environmental policies produce unintended welfare consequences.

\subsection{Multi-level governance and regulatory competition}

A large literature examines the design of carbon pricing instruments, including revenue recycling \citep{Marron2014, timilsina2022carbon} and price-based versus quantity-based approaches \citep{weitzman2017voting}. Research on climate clubs emphasizes the difficulty of sustaining international cooperation without enforcement and highlights the pervasive problem of free-riding \citep{carraro1993strategies, barrett1994self, Nordhaus2015}, with recent work suggesting that current mitigation efforts are insufficient to meet temperature targets under plausible uncertainty \citep{Nordhaus2018, Stiglitz2019}.

When multiple carbon pricing instruments apply simultaneously, their interaction depends on design, timing and location in the value chain. \citet{perino2025overlapping} develop a general framework classifying overlapping policies according to their interaction with a cap-and-trade system, and \citet{wu2026policy} empirically confirm this pattern across more than 100 countries. Related work shows that overlapping instruments can reinforce or undermine an underlying carbon market depending on the degree of coordination between jurisdictions \citep{goulder2011challenges, Baylis2013}. These contributions provide conceptual foundations but analyse overlapping instruments at a macro or sectoral level without modelling individual firm responses.

The environmental federalism literature is closely related. Classical analyses of multi-level government \citep{oates1972fiscal, gordon1983optimal, keen2002does} have been extended to environmental policy in settings with cross-border spillovers and heterogeneous abatement costs \citep{williams2012growing, banzhaf2012fiscal, coria2021interjurisdictional}. \citet{goulder2012unintended} identify severe leakage when regional initiatives are nested within uniform federal standards, since strict regional limits relax the binding federal constraint and allow firms to shift high-emission activity to unregulated markets. \citet{coria2021interjurisdictional} show that a federal tax on a local pollutant can lead local authorities to raise their own rates, increasing rather than rationalising total taxation. More broadly, \citet{kennedy1994equilibrium} and \citet{barrett1994strategic} show that governments strategically set environmental taxes to shift rents toward domestic firms under imperfect competition. We find a similar pattern, with local regulators using emission charges partly to protect domestic carriers and extract surplus from foreign passengers.

The transportation literature has focused primarily on congestion pricing rather than environmental regulation. \citet{pels2004economics} show that a Pigouvian tax can be welfare-reducing when market power is present, with optimal tolls potentially negative under coordination and excessive under tax competition. \citet{silva2014airline} extend this to airline route competition and confirm that regulatory competition generates welfare-reducing deviations from the social optimum. Our setting differs in two respects. The tax base is not a set of independent facilities but a network industry in which firms connect multiple jurisdictions, so regulators influence each other through the network decisions of multi-market firms in addition to standard tax-base effects. We also consider regional and global regulators jointly in an inherently transnational sector, capturing the consequences for firms that must navigate jurisdiction-specific welfare maximization alongside global environmental objectives.

\subsection{Market distortions in network industries}

Network industries exhibit three structural features that complicate environmental regulation. Hub-and-spoke networks and limited infrastructure generate market power, allowing firms to extract surplus through prices that deviate from competitive benchmarks \citep{borenstein1989hubs, brueckner1992fare}. Since output in oligopolistic markets falls below the competitive level, taxing emissions further restricts output and may reduce welfare even when addressing a genuine externality \citep{barnett1980pigouvian, brueckner2002airport}. These distortions interact with positive frequency externalities. The Mohring effect \citep{mohring1972optimization} implies that higher frequency reduces schedule delay and raises customer utility, so charges that reduce frequency impose welfare costs beyond the direct price effect. Multi-jurisdictional networks also expose firms and regulators to carbon leakage, with charges in one region reallocating traffic and emissions toward less regulated markets \citep{carbone2013linking, Perino2019}. \citet{deJong2022fleet} document a related form of leakage in which faster fleet upgrades on short-haul routes subject to the EU ETS are partially offset by slower upgrades on long-haul routes.

Our contribution integrates these strands. We model first-stage regulatory decisions by multiple authorities optimizing jurisdiction-specific welfare subject to global environmental damages, together with a second-stage model of oligopolistic firm competition in which carriers respond through pricing, capacity and fleet decisions. Endogenising both the regulatory and the firm-level decisions clarifies how market distortions and network effects shape the welfare outcomes of carbon pricing, and connects firm-response and regulatory-competition literatures that have developed largely independently.

\section{Analytical Model}
\label{sec:analytical_model}

We develop a stylized analytical framework that isolates the key mechanisms through which firms respond to environmental charges and how these responses interact with regulatory design. The model features two symmetric airlines serving a single market, up to three regulators operating at local and global levels, and passengers whose willingness to pay is also a function of service frequency. We state the main assumptions below, then derive the policy instrument and regulatory best responses (\ref{subsec:generic_instrument}--\ref{subsec:scenarios}), establish equilibrium and welfare rankings (\ref{subsec:rankings}--\ref{subsec:welfare}), and extend the analysis to asymmetric competition (\ref{subsec:asymmetry}). The assumptions are as follows:

\begin{description}
  \item [\textit{Assumption 1:}] Airlines compete in a Cournot fashion. Specifically, airlines choose the quantity they want to supply, and the fares are derived accordingly from the inverse demand function. This is in line with the wide majority of supply chain industries and in particular with the aviation sector \citep{pels2004economics,silva2014airline}
  \item [\textit{Assumption 2:}] The airline $i$ cost function is composed of a variable cost $c_p$ depending on the number of passengers $q_i$, a variable cost $c_f$ based on the service frequency $f_i$, and a fixed cost $F$:
        \begin{equation}
          c_p q_i + c_f f_i + F
          \label{eq:cost function cournot}
        \end{equation}
  \item [\textit{Assumption 3:}] Airlines operate at a constant load factor, implying that service frequency depends on the seat capacity $s$ offered on a flight:
        \begin{equation}
          f_{i} = \frac{q_{i}}{s}
          \label{eq:load_factor}
        \end{equation}
\end{description}

This assumption implies that an increase in the number of passengers served by an airline leads directly to higher service frequency. This creates a positive frequency externality for passengers because the more travelers on a route, the higher the frequency, which benefits all users of that route. The Mohring effect \citep{mohring1972optimization} plays a central role in the model by separating the welfare-maximizing charge from the standard Pigouvian benchmark.

\begin{description}
  \item[\textit{Assumption 4:}] The passenger utility function is quadratic in the consumption of flights \citep{singh1984price, adler2016regulating} and the quality measure depends on frequency:
        \begin{equation}
          U(q_1, q_2) = \sum_{i=1}^2[\alpha + \beta \ln(1+f_i)] q_i - \frac{1}{2} (q_1 + q_2)^2
          \label{eq:utility}
        \end{equation}
        which, when differentiated with respect to $q_i$, leads to an inverse demand function:
        \begin{equation}
          \frac{\partial U}{\partial q_i} = \left[\alpha + \beta \ln(1 + f_i)\right] + \frac{\beta q_i}{s + q_i} - (q_i + q_j) = \alpha + \beta\, \phi(q_i) - (q_i + q_j) = p_i,
          \label{eq:inverse_demand}
        \end{equation}
        where $\alpha > 0$, $\beta > 0$ and $s > 0$ and:
        \begin{equation}
          \phi(q_i) \equiv \ln \left(1 + \frac{q_i}{s} \right) + \frac{q_i}{s + q_i}
          \label{eq:phi_function}
        \end{equation}
\end{description}

The function $\phi(q_i)$ captures the combined effect of frequency on willingness to pay. Its first component, $\ln(1 + q_i/s)$, reflects the direct utility passengers derive from higher frequency, while the second component, $q_i/(s + q_i)$, captures the marginal frequency benefit that the airline can appropriate through fares. The parameter $\beta$ governs the strength of the frequency effect. When $\beta$ is large relative to $s$, frequency externalities are strong and the case for taxing emissions is weakened because the social benefit from additional flights is high.

\begin{description}
  \item [\textit{Assumption 5:}] We assume that airlines operate a single type of aircraft in our model. We will relax this assumption in the case study, allowing for environmental improvements in the fleet composition.
\end{description}

We define the properties of $\phi$ and the consumer surplus expression in Appendix~\ref{app:aux}.

\subsection{Policy instrument}
\label{subsec:generic_instrument}

We model carbon policy as a round-trip instrument $T$ that enters the airline's marginal cost $m(T)$ additively through frequency:
\begin{equation}
  m(T) := c_p + \frac{c_f + T}{s}.
  \label{eq:mT}
\end{equation}
Given $T$, airline $i$'s profit is:
\begin{equation}
  \pi_i(q_i, q_j; T) = \big[p_i(q_i, q_j) - m(T)\big]q_i - F.
  \label{eq:profit_analytics}
\end{equation}

\begin{lemma}[Symmetric Cournot equilibrium]
  \label{lem:cournot}
  Let $T$ be any scalar instrument and define the Cournot first-order condition as:
  \begin{equation}
    G_i(q_i, q_j; T) := \alpha - m(T) - q_j + \beta\, \phi(q_i) + \beta q_i \phi'(q_i) - 2q_i = 0.
    \label{eq:Gi}
  \end{equation}
  Under symmetry ($q_i = q_j = q$), the equilibrium quantity $q^*(T)$ is uniquely determined by the first-order condition:
  \begin{equation}
    G(q; T) := \alpha - m(T) + \beta \phi(q) + \beta q \phi'(q) - 3q = 0.
    \label{eq:Gsym}
  \end{equation}
  Furthermore, provided standard stability conditions hold (see Appendix \ref{app:cournot}), the equilibrium output $q^*(T)$ is strictly decreasing in $T$, while the equilibrium price $p^*(T)$ is strictly increasing in $T$.
\end{lemma}

\subsection{Regulatory settings and best responses}
\label{subsec:scenarios}

We consider four regulatory settings in the analytical model, summarized in Table~\ref{tab:scenarios}, distinguished by the interactions between local and global authorities. A fifth benchmark, differentiated global regulation, is introduced in the case study once regional heterogeneity makes region-specific charges meaningful. Throughout the analysis, we distinguish coordination from instrument flexibility. Coordination refers to whether one authority internalizes emissions and surplus across regions. Instrument flexibility refers to whether the authority can set one common charge or region-specific charges. These dimensions need not coincide: uniform global regulation is coordinated but inflexible, decentralized regional regulation is flexible but non-cooperative, differentiated global regulation combines coordination with regional flexibility, and overlapping regulation combines regional and global instruments. We assume that consumer surplus and environmental damage, $e$, are shared equally between the two jurisdictions. This assumption is consistent with IPCC assessments that indicate comparable climate vulnerability across North American and European regions \citep{portner2022climate}.

\begin{table}[ht]
  \centering
  \caption{Regulatory scenarios and benchmark labels}
  \label{tab:scenarios}
  \begin{tabular}{lll}
    \hline
    Scenario & Description                       & Total Instrument ($T$)                  \\ \hline
    1        & Baseline                          & $T^{(1)} = 0$                           \\
    2        & Decentralized regional regulation & $T^{(2)} = \tau_L + \tau_{L'}$          \\
    3        & Uniform global regulation         & $T^{(3)} = \tau_G$                      \\
    4        & Overlapping regulation            & $T^{(4)} = \tau_L + \tau_{L'} + \tau_G$ \\ \hline
  \end{tabular}
\end{table}

The welfare functions for a local regulator ($SW_L$) and global regulator ($SW_G$) are defined as:
\begin{equation}
  SW_L = \tfrac{1}{2} CS(q_1, q_2) + \pi_1 + \tau_L\tfrac{(q_1 + q_2)}{s} - \tfrac{1}{2}\tfrac{e(q_1 + q_2)}{s}
  \label{eq:SWL}
\end{equation}
\begin{equation}
  SW_G = CS(q_1, q_2) + \pi_1 + \pi_2 + \tau_G \tfrac{(q_1 + q_2)}{s} - \tfrac{e(q_1 + q_2)}{s}
  \label{eq:SWG}
\end{equation}

We can now define the regulator's best response in each scenario. Let $e > 0$ denote the marginal environmental damage per round trip.

\begin{proposition}[Local and global best responses]
  \label{prop:best_responses}
  Evaluating the best response for symmetric output $q$:

  \smallskip
  \noindent (i) The optimal carbon tax for the local regulator in Scenarios 2 and 4 is:
  \begin{equation}
    \tau_L^*(q) = \tau_{L'}^*(q) = \frac{1}{2}\left[e + 2sq - \frac{2\beta s^3 q}{(s+q)^3}\right]
    \label{eq:tauL}
  \end{equation}

  \smallskip
  \noindent (ii) The optimal carbon tax for the global regulator in Scenarios 3 and 4 is:
  \begin{equation}
    \tau_G^*(q) = e - sq\bigl[1 - \beta\phi'(q)\bigr] = e - sq + \frac{\beta s q(2s+q)}{(s+q)^2}.
    \label{eq:tauG}
  \end{equation}
\end{proposition}

The optimal charges in Equations~\ref{eq:tauL} and~\ref{eq:tauG} each contain three components. Both the local and global regulators set charges that reflect the environmental damage, $e$, however airlines already restrict output below competitive levels, so an additional tax exacerbates this distortion, leading regulators to set the carbon tax below the social cost of emissions. The Mohring effect, captured by the $\beta$-dependent terms, partially offsets this downward pressure because a reduction in frequency lowers consumer surplus for the remaining passengers. The optimal carbon tax is therefore set below the social cost of emissions but remains positive for sufficiently high environmental damage estimations.

\begin{corollary}[Tax ranking]
  \label{cor:tax_rank}
  Assuming $\beta < s/2$, then for all $q > 0$:
  \begin{equation}
    2\tau_L^*(q) > \tau_G^*(q),
    \qquad
    T^{(4)}(q) = 2\tau_L^*(q) + \tau_G^*(q) > 0.
    \label{eq:tax_rank}
  \end{equation}
\end{corollary}

This ranking shows that decentralized regional regulators will collectively set a higher tax relative to the coordinated uniform benchmark. While each local authority internalizes the environmental damage from departures within its jurisdiction, it only accounts for the welfare losses borne by its own residents and airlines. Consequently, uncoordinated local policies lead to excessive output restrictions, a distortion that is further compounded under the overlapping charges of Scenario 4.

\subsection{Equilibrium rankings}
\label{subsec:rankings}

Let $q^{(\sigma)*}$ and $p^{(\sigma)*}$ denote the symmetric equilibrium quantity and price in scenario $\sigma \in \{1,2,3,4\}$, with policy instruments evaluated at their fixed points.

\begin{proposition}[Ranking when $\tau_G^*(q) \geq 0$]
  \label{prop:rank_pos}
  Assume Lemma~\ref{lem:cournot} and suppose $\tau_G^*(q^{(3)*}) \geq 0$. Then:
  \begin{align}
     & 0 = T^{(1)} \le T^{(3)*} < T^{(2)*} < T^{(4)*}, \label{eq:T_rank_pos} \\
     & q^{(4)*} < q^{(2)*} < q^{(3)*} \leq q^{(1)*}, \label{eq:q_rank_pos}   \\
     & p^{(4)*} > p^{(2)*} > p^{(3)*} \geq p^{(1)*}. \label{eq:p_rank_pos}
  \end{align}
  If $\tau_G^*(q^{(3)*})>0$, all inequalities are strict.
\end{proposition}

These rankings have a direct environmental interpretation. Since emissions are proportional to output in our framework, the quantity ordering in Eq.~(\ref{eq:q_rank_pos}) also ranks the scenarios by total emissions. Uniform global regulation (Scenario~3) delivers the socially optimal emission level in the symmetric analytical benchmark, where one common charge is sufficient. Decentralized regional regulators (Scenario~2) over-tax and therefore reduce emissions below this optimum, at the cost of excessive output restriction. Overlapping regulation (Scenario~4) produces the most severe output contraction and the lowest emissions, but this additional abatement comes at a cost in welfare that exceeds its environmental benefit.

\subsection{Welfare rankings}
\label{subsec:welfare}

To distinguish aggregate welfare from the individual objective of regulators, we define global (network) welfare as:
\begin{equation}
  \mathcal{W}(q_1,q_2;e) = U(q_1, q_2) - \sum_{i=1}^2 \bigl(c_p q_i + c_f f_i + F\bigr) - e\,\tfrac{q_1+q_2}{s}.
  \label{eq:W_def}
\end{equation}
We note that the revenues collected from the carbon charge are a pure transfer from airline profits to government surplus and therefore cancel out. Under symmetry ($q_1 = q_2 = q$), using $f_i = q/s$:
\begin{equation}
  \mathcal{W}(q;e) = 2(\alpha - c_p)q + 2\beta q\ln\!\Bigl(1 + \tfrac{q}{s}\Bigr)- 2q^2 - \tfrac{2(c_f+e)}{s}q - 2F.
  \label{eq:W_sym}
\end{equation}

\noindent It is convenient to decompose welfare as:
\begin{equation}
  \mathcal W(q; e) = \mathcal W_0(q) - \frac{2e}{s}q
  \label{eq:W_decomp}
\end{equation}
where $\mathcal W_0(q):=\mathcal W(q;0)$ denotes baseline welfare in the absence of environmental externalities.

\begin{lemma}[Concavity of global welfare]
  \label{lem:welfare_concavity}
  For $\beta < s/2$ and $\alpha > c_p + (c_f + e)/s$, then $\mathcal W(q; e)$ is strictly concave in $q \geq 0$ and the social optimum (SO),
  $q^{SO}(e)$, is characterized by:
  \begin{equation}
    \alpha + \beta \phi\big(q^{SO}(e)\big) - 2q^{SO}(e) = c_p + \frac{c_f + e}{s}.
    \label{eq:qSO}
  \end{equation}
\end{lemma}

\begin{proposition}[Uniform global regulation leads to the social optimum in1 the symmetric benchmark]
  \label{prop:scenario3_SO}
  In Scenario~3, the uniform global regulation implements $q^{SO}(e)$ with a per round-trip charge:
  \begin{equation}
    \tau_G^{(3)*}(e) = e - s q^{SO}(e) \Big[1 - \beta \phi'\big(q^{SO}(e)\big)\Big].
    \label{eq:tauG_SO}
  \end{equation}
  Consequently, $\mathcal W^{(3)*}(e) = \max_{q \geq 0}\mathcal W(q; e)$ and Scenario~3 dominates the other analytical regimes from a global welfare perspective in the symmetric benchmark.
\end{proposition}

\begin{proposition}[Welfare ranking when the global tax is positive]
  \label{prop:welfare_pos}
  For $\beta < s/2$ and $\tau_G^{(3)*}(e) > 0$, then:
  \begin{equation}
    \mathcal{W}^{(3)*}(e) > \mathcal{W}^{(2)*}(e) > \mathcal{W}^{(4)*}(e),
    \qquad
    \mathcal{W}^{(3)*}(e) > \mathcal{W}^{(1)*}(e).
    \label{eq:W_rank_pos}
  \end{equation}
  Comparisons involving Scenario~1 versus Scenarios~2 and~4 depend on $e$.
  Let
  \begin{equation}
    \bar{e}_{21} := \inf\bigl\{ e \geq 0 : \mathcal{W}^{(2)*}(e) \geq \mathcal{W}^{(1)*}(e) \bigr\},
    \qquad
    \bar{e}_{41} := \inf\bigl\{ e \geq 0 : \mathcal{W}^{(4)*}(e) \geq \mathcal{W}^{(1)*}(e) \bigr\}.
    \label{eq:threshold_defs}
  \end{equation}
  Then $\bar e_{21} \leq \bar e_{41}$ and the remaining welfare ranking among Scenarios~1,2,4 is
  \begin{align*}
    \begin{array}{ll}
      \mathcal{W}^{(1)*} > \mathcal{W}^{(2)*} > \mathcal{W}^{(4)*}
       & \text{if } e < \bar{e}_{21},                   \\[4pt]
      \mathcal{W}^{(2)*} \geq \mathcal{W}^{(1)*} > \mathcal{W}^{(4)*}
       & \text{if } \bar{e}_{21} \leq e < \bar{e}_{41}, \\[4pt]
      \mathcal{W}^{(2)*} > \mathcal{W}^{(4)*} \geq \mathcal{W}^{(1)*}
       & \text{if } e \geq \bar{e}_{41}.
    \end{array}
  \end{align*}
  Welfare thresholds $\bar e_{21}$ and $\bar e_{41}$ are specified in Appendix~\ref{app:thresholds}.
\end{proposition}

The welfare ordering reveals a counter-intuitive result since the addition of a regulatory layer does not necessarily improve environmental outcomes. Overlapping regulation (Scenario~4) never outperforms local-only regulation (Scenario~2) when environmental damages are high because the additional global charge increases the excess taxation problem caused by decentralized regional regulators. When damages are low, the comparison between these two scenarios becomes ambiguous and depends on whether the role of the global regulator (pushing toward the social optimum) can offset the distortions introduced by overlapping taxes. The thresholds $\bar{e}_{21}$ and $\bar{e}_{41}$ define the regions in which local regulation improves upon no regulation at all, providing a formal basis for assessing when fragmented carbon policies are welfare-improving versus welfare-reducing.

\begin{lemma}[Threshold for the sign of the global instrument]
  \label{lem:e_threshold_neg}
  Let $q^{(1)*}$ denote the baseline equilibrium (under $T^{(1)} = 0$) and define:
  \begin{equation}
    \bar e_G := s q^{(1)*} \Big[1 - \beta \phi'\big(q^{(1)*}\big)\Big].
    \label{eq:dbarG}
  \end{equation}
  Then $\tau_G^{(3)*}(e) < 0 \Longleftrightarrow e < \bar e_G$, $\tau_G^{(3)*}(e) = 0 \Longleftrightarrow e = \bar e_G$, and $\tau_G^{(3)*}(e) > 0 \Longleftrightarrow e > \bar e_G$. Moreover, $q^{(3)*}(e) \gtrless q^{(1)*}$ iff $e \lessgtr \bar e_G$.
\end{lemma}

\begin{proposition}[Rankings and welfare when the global tax is negative]
  \label{prop:subsidy_regime}
  Assume $\beta < s/2$ and $e < \bar e_G$ (Lemma~\ref{lem:e_threshold_neg}), so that $\tau_G^{(3)*}(e) < 0$. Then the instruments and quantities satisfy
  \[
    T^{(3)*}(e) < 0 = T^{(1)} < T^{(2)*}(e),
    \qquad
    0 < T^{(4)*}(e),
  \]
  and
  \begin{equation*}
    q^{(2)*}(e) < q^{(1)*} < q^{(3)*}(e),
    \qquad
    q^{(4)*}(e) < q^{(1)*} < q^{(3)*}(e).
  \end{equation*}
  with the reverse ordering for prices. In global welfare:
  \begin{equation}
    \mathcal{W}^{(3)*}(e) > \mathcal{W}^{(1)*}(e) > \max\bigl\{\mathcal{W}^{(2)*}(e),\ \mathcal{W}^{(4)*}(e)\bigr\}.
    \label{eq:W_rank_neg}
  \end{equation}
  Finally, $\mathcal W^{(4)*}(e) \gtrless \mathcal W^{(2)*}(e)$ iff $q^{(4)*}(e) \gtrless q^{(2)*}(e)$,	equivalently iff $\tau_G^*(q^{(2)*}(e)) \lessgtr 0$.
\end{proposition}

\begin{table}[!ht]
  \centering
  \renewcommand{\arraystretch}{0.5}
  \caption{Rankings across scenarios (1 = lowest, 4 = highest).}
  \label{tab:qualitative_rankings}
  \begin{tabular}{lcccc}
    \toprule
                           & \multicolumn{4}{c}{Scenario}                                                 \\
    \cmidrule(lr){2-5}
                           & S1: Baseline                 & S2: Two locals & S3: Global & S4: Overlapping \\
    \midrule
    \multicolumn{5}{l}{\textit{Regime A:} $\tau_G^{(3)*}(e) \geq 0$ (Proposition~\ref{prop:welfare_pos})} \\
    Quantity $q$           & 4                            & 2              & 3          & 1               \\
    Price $p$              & 1                            & 3              & 2          & 4               \\
    Instrument $T$         & 1                            & 3              & 2          & 4               \\
    Welfare $(\mathcal W)$ & 1-2-3                        & 2-3            & 4          & 1-2             \\
    \addlinespace
    \multicolumn{5}{l}{\textit{Regime B:} $\tau_G^{(3)*}(e) < 0$ (Proposition~\ref{prop:subsidy_regime})} \\
    Quantity $q$           & 3                            & 1-2            & 4          & 1-2             \\
    Price $p$              & 2                            & 3-4            & 1          & 3-4             \\
    Instrument $T$         & 2                            & 3-4            & 1          & 3-4             \\
    Welfare $(\mathcal W)$ & 3                            & 1-2            & 4          & 1-2             \\
    \bottomrule
  \end{tabular}
\end{table}

\subsection{Asymmetric competition}
\label{subsec:asymmetry}

We next extend the analytical framework with a numerical model that introduces two sources of regional asymmetry. First, each region includes a local market served exclusively by the regional airline, as shown in Figure~\ref{fig:network_asym}. Second, flight distances differ across regions through a parameter $\zeta \geq 1$, which scales the distance flown in one region relative to the other. This distance asymmetry first affects airline operations, as longer flights increase both operating costs and emissions. These higher costs are then transmitted to consumers through equilibrium fares and service frequencies, particularly in the monopoly local market. In the monopoly local market, these higher costs are borne by passengers through the monopolist's pricing and service decisions. The numerical model therefore allows asymmetry in distances to propagate into asymmetry in costs, pollution, market outcomes, and regulatory incentives.

To capture these features, we modify the airline profit function to include the two local monopoly markets and introduce $\zeta$ as a distance shifter for flights in one region. We then use the numerical model to examine whether the welfare ranking derived in the symmetric analytical setting continues to hold once regions differ in both market structure and flight distance.

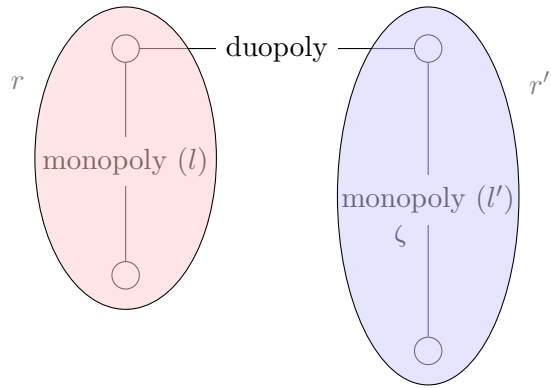
\begin{figure}[!htb]
  \centering
  \caption{Asymmetric network}
  \vspace{0.2 cm}
  \begin{tikzpicture}
    \node[circle,draw=black] (Point1) at (0, 4.45) {};
    \node[circle,draw=black] (Point2) at (4, 4.45) {};
    \node[circle,draw=black] (Point3) at (0, 1.45) {};
    \node[circle,draw=black] (Point4) at (4, 0.45) {};

    \node (P) at (3.65, 2) {$\zeta$};

    \draw (Point1) -- (Point2) node [midway, fill=white] {duopoly};
    \draw (Point1) -- (Point3) node [midway, fill=white] {monopoly $(l)$};
    \draw (Point2) -- (Point4) node [midway, fill=white] {monopoly $(l')$};

    \draw [fill=red!20, fill opacity=0.5]
    (0,3) ellipse (1.2cm and 2cm) node[left] at (-1.2,4) {$r$};

    \draw [fill=blue!20, fill opacity=0.5]
    (4,2.5) ellipse (1.2cm and 2.5cm) node[right] at (5.2,4) {$r'$};
  \end{tikzpicture}
  \label{fig:network_asym}
\end{figure}

\begin{table}[!htb]
  \centering
  \caption{Parameters for the numerical analysis.}
  \begin{tabular}{cccccccc}
    \hline
    $\alpha$ & $\beta$ & $c_p$ & $c_f$   & F         & s   & e      \\
    \hline
    5,000    & 140     & 100   & 220,000 & 3,000,000 & 300 & 88,000 \\
    \hline
  \end{tabular}
  \label{tab:parameters_numerics}
\end{table}

Table~\ref{tab:parameters_numerics} reports the parameter values used in the numerical analysis. These values are calibrated to an average long-haul round-trip flight. Total $CO_2$-equivalent emissions amount to 440 tons and are valued at \euro~200 per ton. For the local monopoly markets, parameter values are set to one-third of the corresponding long-haul values, except for $\alpha$. This captures the lower operating costs of regional flights, the use of smaller aircraft, and the lower contribution of frequency to passenger utility, while still allowing for relatively high passenger volumes. In the overlapping-regulators scenario, we further assume that revenues collected by the global regulator are redistributed equally across the two regional authorities, thereby preserving revenue neutrality at the global level.

\begin{figure}[!htb]
  \centering
  \caption{Numerical results under asymmetric competition}
  \label{fig:numerics}

  \begin{subfigure}[t]{0.49\textwidth}
    \centering
    \includegraphics[width=\textwidth]{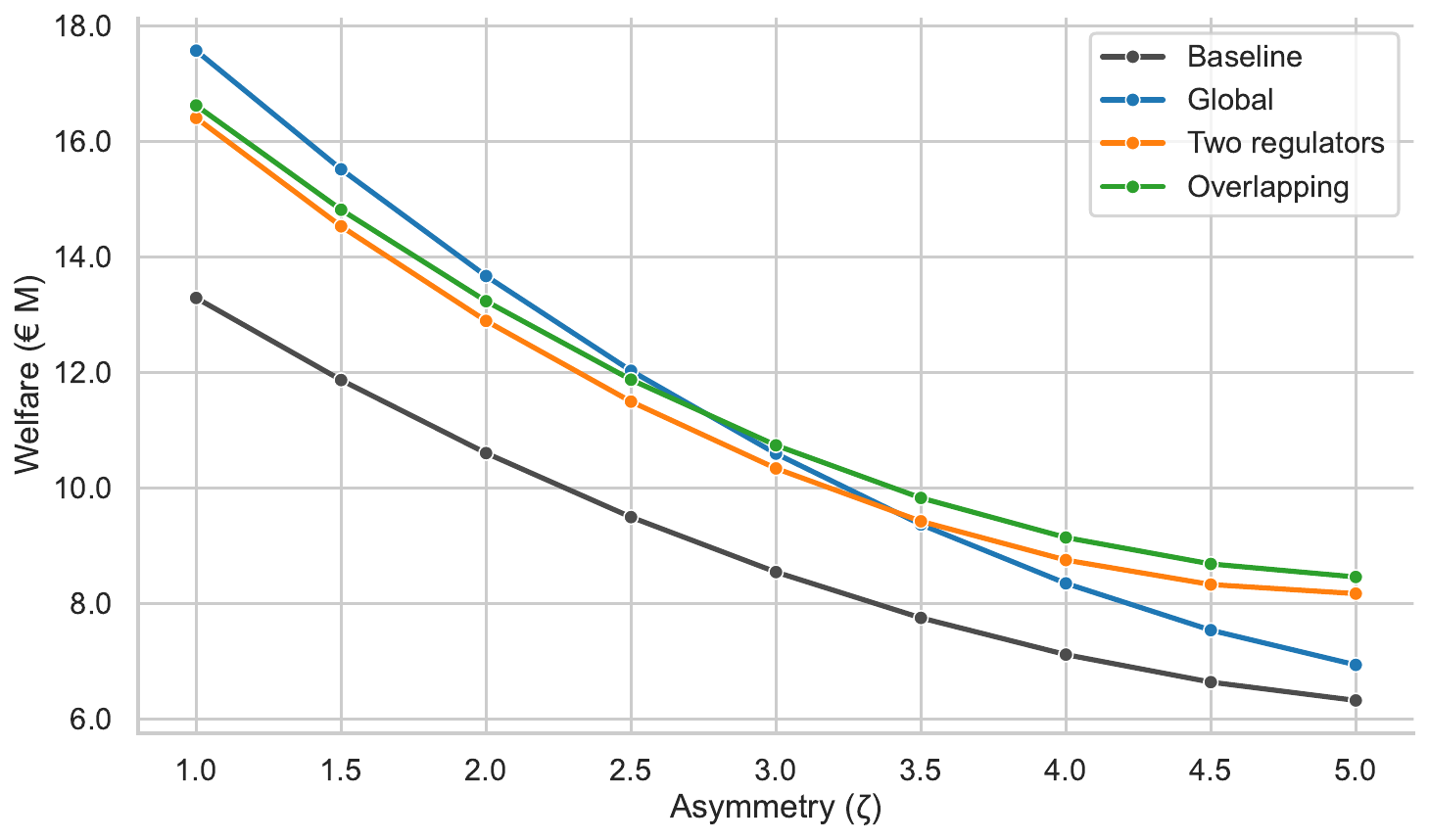}
    \caption{Social welfare}
    \label{fig:numerics_welfare}
  \end{subfigure}
  \hfill
  \begin{subfigure}[t]{0.49\textwidth}
    \centering
    \includegraphics[width=\textwidth]{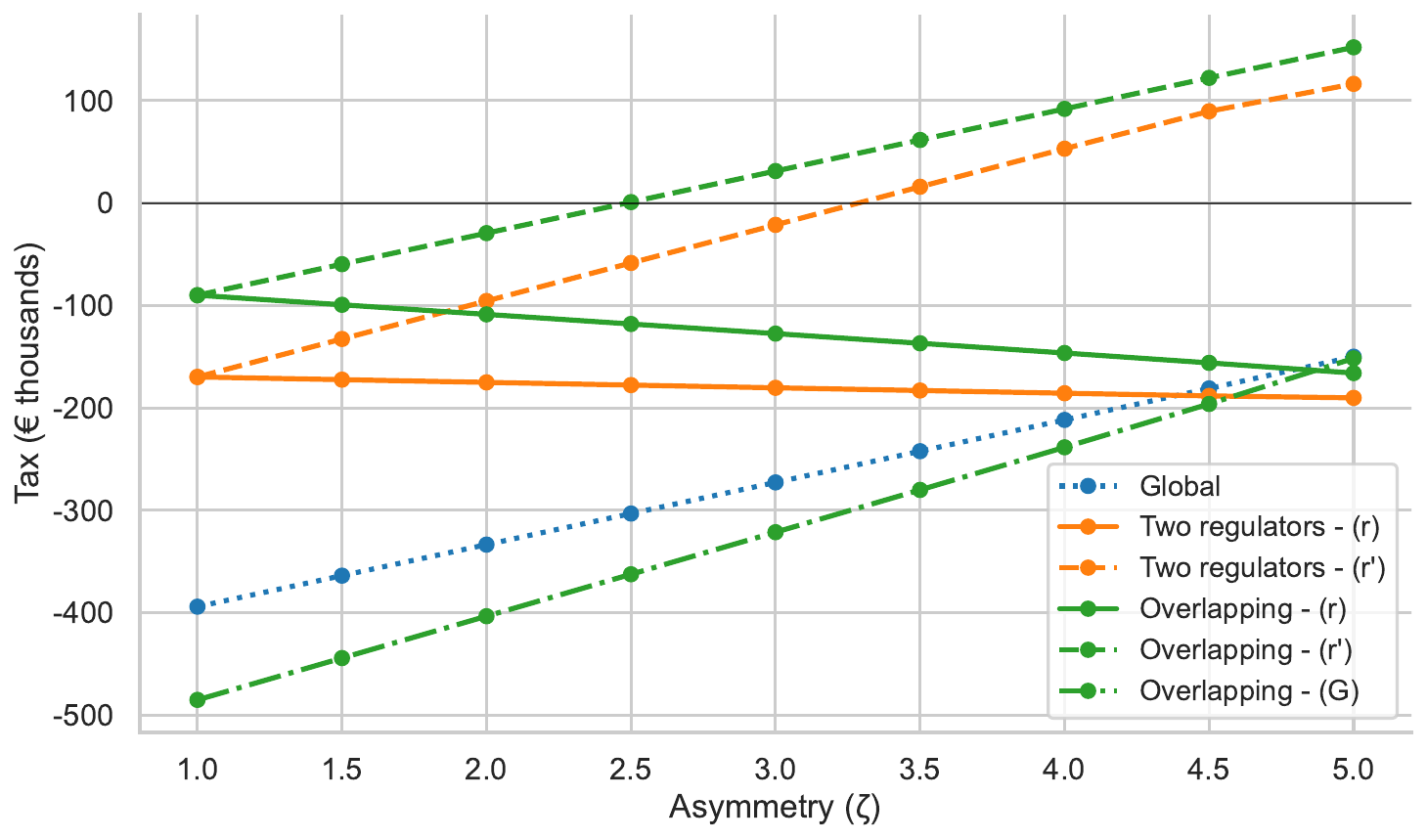}
    \caption{Equilibrium charges}
    \label{fig:numerics_charges}
  \end{subfigure}
\end{figure}

Figure~\ref{fig:numerics} reports the welfare results and the corresponding equilibrium charges. Social welfare declines as asymmetry increases. For low levels of asymmetry, the global-regulator scenario yields the highest welfare, followed by overlapping regulation. As asymmetry rises, this ranking reverses. Welfare under the global regulator falls below welfare under both overlapping regulation and competing regional regulation. Although two decentralized regional regulators continue to perform worse than the other regulated settings at low levels of asymmetry, they dominate the no-intervention benchmark throughout.

The charge patterns on round trip flights clarify the mechanism behind this reversal. Negative charges correspond to subsidies. At low levels of asymmetry, both regional regulators subsidize their aviation sectors. As $(\zeta)$ increases, the regulator in the long-distance region raises its charge and eventually switches from a subsidy to a positive tax, while the regulator in the short-distance region continues to subsidize. A similar pattern arises under overlapping regulation, where the regional instrument in the long-distance region becomes positive as asymmetry increases. Thus, the welfare reversal is driven by the increasing value of differentiated instruments.

The underlying mechanism is straightforward. As asymmetry rises, the region with longer distances becomes both more polluting and more costly to serve, and the regional airline raises fares and reduces service in its monopoly market to recover part of the higher distance-induced cost. A uniform global instrument can vary with the degree of asymmetry, but it cannot differentiate across regions at a given level of $\zeta$. It therefore applies a single charge to markets whose optimal corrective instruments increasingly diverge. Regimes with multiple regulatory layers can respond more flexibly. For high asymmetry, welfare under decentralized regional and overlapping regulation therefore exceeds welfare under uniform global regulation, with overlapping regulation performing best. The reason is not that additional regulatory layers are desirable per se, but that a single uniform instrument is too coarse to accommodate cross-region heterogeneity in costs, emissions, and monopoly distortions.

\section{Extended Model}\label{sec:extended_model}
In this section, we transition from the stylized Cournot framework to a choice-based numerical model. To accommodate passenger heterogeneity, varying fleet efficiencies and hub-and-spoke routing, we adopt the nested logit demand framework of \citet{adler2026estimating}. This richer environment preserves the two-stage game: regulators first set competitive emission charges and airlines subsequently compete on prices and service frequencies. The discrete choice formulation allows us to map how local carbon taxes propagate through a network via operational adjustments, such as upgrading fleets or reallocating high-emission assets. Table~\ref{tab:notation} summarizes the notation used in the numerical model.

\begin{table}[!t]
  \caption{Notation.}
  \resizebox{\textwidth}{!}{%
    \input{Notation}}
  \label{tab:notation}
\end{table}

\subsection{Network design}

We define a hub-and-spoke network, $G(\mathcal{N},\mathcal{K})$, in which each carrier has a hub base in its home country. Each airport is represented as a node in the set $\mathcal{N}$. The hub is connected to its spokes via directed legs in the set $\mathcal{K}$, enabling indirect connectivity between spoke airports through the hub. In this setting, we do not consider code-sharing or interlining.

Given this network configuration, airlines may be subject to different climate policy regimes imposed by regulators. We analyze alternative regulatory settings, including a multi-regulator environment with competing jurisdictions and a single-regulator regime. In our model, regulators levy a carbon tax per ton of $CO_2$ emitted by flights departing from airports under their authority. The tax may vary by aircraft type, thereby incentivizing the deployment of more fuel-efficient aircraft.

The sets within the jurisdiction of regulator \textit{r} are defined as follows:
\begin{align*}
  \mathcal{N}^r & = \bigl\{ i^r, j^r \in \mathcal{N} \mid
  i^r, j^r \text{ are nodes regulated by } r \bigr\}      \\
  \mathcal{A}^r & = \bigl\{ a^r \in \mathcal{A} \mid
  a^r \text{ is an airline regulated by } r \bigr\}       \\
  \mathcal{K}^r & = \bigl\{ k^{hr} \in \mathcal{K} \mid
  k^{hr} \text{ is a leg regulated by } r \bigr\}
\end{align*}

\subsection{First stage}
As in the first stage formulated above, regulators compete to maximize social welfare while accounting for environmental damage. This implies that each regulator's aviation sector contributes to total emissions, but the resulting damages are not borne uniformly across regions. Unlike the analytical model, we allow heterogeneous climate-risk exposure through the parameter $\eta_r$. Regulators whose regions are less affected by climate change may therefore have an incentive to free-ride on the emissions reductions driven by regulators in more climate-vulnerable regions.

All $CO_2$ emissions generated by aviation impose a social cost shared across jurisdictions and captured by the social cost of carbon $e$. Since carbon emissions have global environmental impacts, $e$ is assumed to be homogeneous across regions. Consequently, when setting policy, regional regulators internalize not only damages to their own region but also the spillover effects their decisions impose on other regions.

In our game, regulators trade-off environmental externalities against the passenger surplus of their own inhabitants and the profits of their airlines by setting the level of taxation on $CO_2$ in their jurisdiction. Decisions in one jurisdiction affect incentives elsewhere, inducing competition at the regulatory level. The resulting first-stage equilibrium is then taken as given in the second stage, in which airlines compete.

Given these elements, we express social welfare as the sum of four components: passenger surplus, producer profits, environmental damage and governmental income from environmental taxation as follows:
\begin{align}
  \label{eq:first stage}
  \underset{\theta_{r}}{\max} \; sw_{r}
   & = \sum_{i^{r} \in \mathcal{N}^r} \sum_{j \in \mathcal{N}} \sum_{t \in \mathcal{T}}
  \frac{d_{i^{r}j}}{-\beta_{price}} \ln \!\left(1 + \left(\sum_{a' \in \mathcal{A}}
  \exp\!\left(\frac{V_{i^{r}jta'}}{1 - \rho_t}\right)
  \right)^{1 - \rho_t} \right)
  + \sum_{a^r \in \mathcal{A}^r}
  \pi_{a^r} \!\left(f_{k^{hr}va^r}^*, p_{ijta^r}^*,
  \tilde{\alpha}^*_{hva^r}, \theta_{r}\right) \nonumber                                 \\
   & \quad + \sum_{h \in \mathcal{H}} \sum_{k^{hr} \in \mathcal{K}^r}
  \sum_{v \in \mathcal{V}} \sum_{a \in \mathcal{A}}
  \varepsilon_{k^{hr}v} f_{k^{hr}va}^* \theta_{r} - \eta_{r} \sum_{h \in \mathcal{H}} \sum_{k^{h} \in \mathcal{K}}
  \sum_{v \in \mathcal{V}} \sum_{a \in \mathcal{A}}
  \varepsilon_{k^{h}v} e f_{k^{h}va}^*
\end{align}

\noindent where $\varepsilon_{k^{h}v} = g_{k^{h}}\phi_{hv}co_{2}$ denotes the tons of $CO_2$ produced on flight leg $k^{h}$ when operated with aircraft version $v$. In Eq. (\ref{eq:first stage}) the first term represents consumer surplus, expressed as the log-sum of the utility of passengers departing from airports under the jurisdiction of the regulator\footnote{Given that the consumer surplus considers only passengers departing from airports in the regulatory jurisdiction, the surplus of transatlantic passengers is shared equally between the two jurisdictions under a round-trip assumption.}. The second term captures the profits of airlines based in the jurisdiction, the third term is government revenue from the carbon charge levied on emissions generated on regulated legs and the final term represents the share $\eta_r$ of total emissions borne by the regulator.

The decision variable for the regulatory entity is the per-ton carbon charge $\theta_{r}$ applied to emissions from flights departing within its jurisdiction, anticipating airlines’ responses in the second stage. No sign restriction is imposed on $\theta_{r}$. In particular, negative values of $\theta_{r}$ represent subsidies to airlines.

\subsection{Market share model}
\label{subsec:marketshare}
We assume that passengers maximize utility when choosing an airline and itinerary. Following \citet{McFadden1974}, utility can be decomposed into a systematic component, $V_{ijsa}$, and a random component, $\epsilon_{ijsa}$, as follows:
\begin{equation}
  U_{ijta} = V_{ijta} + \epsilon_{ijta}, \qquad \forall i,j \in \mathcal{N}, t \in \mathcal{T}, a \in \mathcal{A}
\end{equation}

\noindent In order to represent the two legs that constitute a connecting flight, it is convenient to define the set of legs belonging to an itinerary as:
\begin{equation*}
  \mathcal{K}^* = \left\{ k^* \in \mathcal{K} \mid k^* \text{ is a leg of itinerary } (i,j,a) \right\}
\end{equation*}

\noindent The systematic component is defined as follows:
\begin{align}
  \label{eq:sys_utility}
  V_{ijta} = & \ \beta_{0t} + \beta_{1t} mkt\_dist_{ij} + \beta_{2t} mkt\_dist_{ij}^{2} + \beta_{3t} gdp_{ij} + \beta_{4t} direct_{ija} + \beta_{5t} detour_{ija} \nonumber \\
             & + \beta_{6t} \ln (1 + freq\_min_{ija}) + \beta_{7t} price_{ijta} \qquad \forall\, i,j \in \mathcal{N},\ t \in \mathcal{T},\ a \in \mathcal{A}
\end{align}
\noindent where $mkt\_dist$ is the itinerary distance, $gdp$ is the GDP per capita of the origin and destination cities, $direct$ is a dummy variable that takes value $1$ if the trip is non-stop otherwise $0$, $detour$ is the additional distance of a non-direct trip, $freq\_min$ is the minimum service frequency on all links belonging to the itinerary, and $price$ is the ticket fare. As suggested by \citet{Hansen1990}, using the logarithm of frequency is suitable for capturing the diminishing marginal utility of higher service frequency. This term represents the Mohring effect in the utility function. Moreover, we use the minimum frequency among the legs composing an itinerary, so that the itinerary frequency is determined by the bottleneck leg along the route.

The random component is assumed to be independently and identically distributed according to a Gumbel distribution. Consequently, demand can be split between airlines into their market shares through a nested logit model (NL):
\begin{equation}
  m_{ijta}  = \frac{\exp\left(\frac{V_{ijta}}{1-\rho_t}\right)} {\sum_{a' \in \mathcal{A}}\exp\left(\frac{V_{ijta'}}{1-\rho_t}\right)} \frac{S_{ijt}^{(1-\rho_t)}}{V_{0} + S_{ijt}^{(1-\rho_t)}} \qquad \forall\, i,j \in \mathcal{N},\ t \in \mathcal{T},\ a \in \mathcal{A}
  \label{eq:marketShare}
\end{equation}
where $S_{ijt} = \sum_{a' \in \mathcal{A}}\exp\!\left(\frac{V_{ijta'}}{1-\rho_t}\right)$, $V_0$ is the utility associated with the outside option of not flying, specific to each market, and $\rho_{t}$ captures taste heterogeneity. In contrast to purely theoretical settings where competition is often modeled in a Cournot fashion, the nested logit market-share formulation enables a more realistic representation of competitive interactions among airlines.

\subsection{Airline operating costs}

According to \citet{swan2006aircraft}, the direct operating cost of the airline can be represented by a cost function that distinguishes between long- and short-haul flights. We modify this function to separate fuel, ownership and residual operating costs for two reasons. First, isolating fuel costs allows us to better capture the increased share of fuel in total operating expenses using more recent fuel price assumptions. Second, ownership costs depend on the number of aircraft deployed, which is a decision variable for the airline, so we model them explicitly in the objective function. We update costs to 2019 values to reflect inflation and changes in operating expenses in recent years, and to convert the original dollar specification into euros. Finally, to account for low-cost carriers, we assume they face operating costs equal to half those incurred by legacy carriers.

\begin{equation}
  \begin{cases}
    c_{k^{h}v} = \bar{p}\phi_{hv} g_{k} + \chi C^s_{kv}
     & \text{if } k^{h} \in \mathcal{K}^s, \\
    c_{k^{h}v} = \bar{p}\phi_{hv} g_{k} + \chi C^l_{kv}
     & \text{if } k^{h} \in \mathcal{K}^l.
  \end{cases}
  \label{eq:costfunction}
\end{equation}

\noindent where $C^s_{kv} = (g_{k} + 722)(s_{kv} + 104){\$}0.019$ and $C^l_{kv} = (g_{k} + 2{,}200)(s_{kv} + 211){\$}0.0115$ denote the short- and long-haul operating costs respectively, and
\begin{align*}
  \mathcal{K}^s & = \bigl\{ k^s \mid k^s \in \mathcal{K}
  \text{ are the short-haul legs served} \bigr\}         \\[6pt]
  \mathcal{K}^l & = \bigl\{ k^l \mid k^l \in \mathcal{K}
  \text{ are the long-haul legs served} \bigr\}
\end{align*}

The first term in Eq. (\ref{eq:costfunction}) captures fuel costs, a major component of airline operating expenses that accounts for roughly 30\% of total costs \citep{iata2019fuel}. Fuel cost is computed as a function of flight distance $g_{k}$, the fuel burn rate per kilometer $\phi_{hv}$ which depends on aircraft type and version, and the fuel price per ton $\bar{p}$ calibrated using the 2019 average price. The second term represents the share $\chi$ of operating costs excluding fuel and ownership costs. We compute $\chi$ by weighting the components in \citet{swan2006aircraft} and removing the shares attributable to fuel and ownership. Under this specification, per-flight operating costs depend only on leg characteristics and the aircraft version used.

\subsection{Ownership costs}

The monthly cost of owning an aircraft is approximated by the equivalent annual capital costs divided by the number of months per year:
\begin{equation}
  o_{hv} = \frac{(\Psi_{hv} - \lambda_{hv})\left( \frac{\ell(1 + \ell)^n}{(1 + \ell)^n - 1} \right) + \lambda_{hv} \ell}{12}
\end{equation}

\noindent where $\Psi$ is the aircraft purchase price, $\lambda$ is the salvage value at the end of the $n$-year time period and $\ell$ is the interest rate. We select four commonly used aircraft models as reference types in our game (see Table~\ref{tab:aircraft})\footnote{Using alternative aircraft models does not lead to materially different capital cost values.}. Purchase prices are based on average list prices reported by aircraft manufacturers \citep{airbus2018,axon2022}. We compute salvage values assuming straight-line depreciation over 30 years and a service life of 20 years. We set the interest rate to $\ell$ = 10\%. Finally, we assume that airlines receive a discount off the aircraft list price. Offering discounts relative to list prices is common in the aviation industry but the discount rate is private information and reflects bilateral negotiations between airlines and aircraft manufacturers.

\subsection{Second stage}
In the second stage, airlines maximize profits given the environmental charges imposed by regulators across jurisdictions. Each airline strategically chooses the service frequency of each aircraft version on each network leg $f_{kva}$, fares on origin-destination routes $p_{ijta}$, and the optimal fleet composition, including the number of aircraft of each type and version $\Tilde{\alpha}_{hva}$. Given this setting, the objective function is modeled as follows:
\begin{align}
  \underset{p_{ijta},\, f_{k^{h}va},\, \tilde{\alpha}_{hva}}{\max}\; \pi_a
   & = \sum_{\substack{i,j \in \mathcal{N} \\ i \neq j}}
  \sum_{t \in \mathcal{T}} d_{ijt}\, m_{ijta}\, p_{ijta} - \sum_{k^{h} \in \mathcal{K}} \sum_{v \in \mathcal{V}}
  c_{k^{h}v} f_{k^{h}va} \nonumber                                                                                                                                      \\
   & \quad - \sum_{k^{hr} \in \mathcal{K}^r} \sum_{v \in \mathcal{V}} \varepsilon_{k^{hr}v} \theta_{r} f_{k^{hr}va} - \sum_{h \in \mathcal{H}} \sum_{v \in \mathcal{V}}
  o_{hv} \tilde{\alpha}_{hva}
  \label{eq:profit}
\end{align}

\noindent where $m_{ijta}$ is the market share function per origin-destination $(i,j)$ per passenger type $t$ per airline $a$ as specified in Eq.~(\ref{eq:marketShare}), $c_{kv}$ denotes the operating costs incurred to serve leg $k$ with aircraft version $v$ as defined in Eq. (\ref{eq:costfunction}), $o_{hv}$ is the monthly ownership cost of aircraft type $h$ version $v$ and $\tilde{\alpha}_{hva}$ is the number of aircraft of type $h$ and version $v$ purchased. Airlines may either internalize regulatory charges or pass them through to passengers via higher fares, and can also adjust service frequency and fleet mix in response.

Before defining the constraints, it is useful to introduce the subsets corresponding to airline itineraries and flight legs
\begin{align*}
  \mathcal{N}^\circ & = \bigl\{ i^\circ, j^\circ \in \mathcal{N} \mid  (i^\circ, j^\circ)
  \text{ itineraries through } k^{h} \bigr\}                                              \\
  \Omega'           & = \bigl\{ \omega' \in \mathcal{K} \mid \omega'
  \text{ first leg of itinerary } (i,j) \bigr\}                                           \\
  \Omega''          & = \bigl\{ \omega'' \in \mathcal{K} \mid \omega''
  \text{ second leg of itinerary } (i,j) \bigr\}
\end{align*}

\noindent We now proceed to define the constraints of the second-stage problem:
\begin{align}
   & m_{ijta}
  = \frac{\exp\left(\frac{V_{ijta}}{1-\rho_t}\right)}
    {\sum_{a' \in \mathcal{A}}\exp\left(\frac{V_{ijta'}}{1-\rho_t}\right)}
  \frac{S_{ijt}^{(1-\rho_t)}}{1 + S_{ijt}^{(1-\rho_t)}}
   &                                          & \forall\, i,j \in \mathcal{N},\ t \in \mathcal{T}
  \label{eq:mktshar}                                                                                  \\
   & z_{ija}
  \leq f_{\omega'va}
   &                                          & \forall\, i,j \in \mathcal{N},\ \omega' \in \Omega'
  \label{eq:Linearizing 1}                                                                            \\
   & z_{ija}
  \leq f_{\omega''va}
   &                                          & \forall\, i,j \in \mathcal{N},\ \omega'' \in \Omega''
  \label{eq:Linearizing 2}                                                                            \\
   & \sum_{i^\circ \in \mathcal{N}^\circ}
  \sum_{j^\circ \in \mathcal{N}^\circ}
  \sum_{t \in \mathcal{T}}
  d_{i^\circ j^\circ t}\, m_{i^\circ j^\circ t}
  \leq \sum_{v \in \mathcal{V}} s_{k^{h}v} f_{k^{h}va}
   &                                          & \forall\, k^{h} \in \mathcal{K}
  \label{eq:capacity constraint}                                                                      \\
   & \sum_{k^{h} \in \mathcal{K}} f_{k^{h}va}
  \leq \bar{f}_h \tilde{\alpha}_{hva}
   &                                          & \forall\, h \in \mathcal{H},\ v \in \mathcal{V}
  \label{eq:utilization}                                                                              \\
   & f_{k^{h}va}
  \in \mathbb{R}^{+}
   &                                          & \forall\, k^{h} \in \mathcal{K},\ v \in \mathcal{V}
  \label{eq:frequencyPositive}                                                                        \\
   & p_{ijta}
  \in \mathbb{R}^{+}
   &                                          & \forall\, i,j \in \mathcal{N},\ t \in \mathcal{T}
  \label{eq:pricePositive}                                                                            \\
   & \tilde{\alpha}_{hva}
  \in \mathbb{R}^{+}
   &                                          & \forall\, h \in \mathcal{H},\ v \in \mathcal{V}
  \label{eq:xPositive}
\end{align}

\noindent Eq. (\ref{eq:mktshar}) embeds the nested logit market-share formulation in the model. Constraints (\ref{eq:Linearizing 1}) and (\ref{eq:Linearizing 2}) linearize the minimum-frequency term for connecting itineraries in the utility specification (Eq. (\ref{eq:sys_utility})). This avoids discontinuities in the market-share function during the solution process. Equation (\ref{eq:capacity constraint}) imposes the capacity constraint, ensuring that the demand served by an airline on each leg does not exceed the seats supplied on that leg, accounting for all itineraries that use it. Constraint (\ref{eq:utilization}) limits the number of flights operated based on average aircraft utilization, distinguishing between long- and short-haul operations. Finally, constraints (\ref{eq:frequencyPositive}) -- (\ref{eq:xPositive}) define the feasible domains of the decision variables.

\subsection{Game-theoretic competition and algorithm}

The competitive interaction between regulators and airlines is modeled as a two-stage game with complete but imperfect information because of the simultaneous decisions within each stage \citep{fudenberg1991game}. The sequence of play is as follows. In the first stage, regulators simultaneously set environmental charges. In the second stage, after observing these charges, airlines simultaneously choose service frequencies, fares, and fleet compositions.\footnote{While the choice of airports in the network are fixed, airlines can stop operating a connection by setting its frequency to zero, allowing them to endogenously decide whether to continue serving a market.} We define the regulator payoffs by a social welfare function and airline payoffs by a profit function. This timing structure allows airlines to respond to policy choices in the first stage.

We solve the game by backward induction, as summarized in Appendix \ref{app:algorithm}.  The algorithm starts by initializing values for the first- and second-stage decision variables. It then computes the second-stage Nash equilibrium by solving the airline problem for all carriers in $\mathcal{A}$, which we refer to as a cycle. Given the resulting second-stage equilibrium, the algorithm updates the first-stage outcomes by solving the regulator problem for all regulators in $\mathcal{R}$. After each first-stage update, a new cycle is computed using the updated charges. The procedure iterates until convergence to a fixed point, which corresponds to a subgame-perfect Nash equilibrium (SPNE) of the overall two-stage game.

The second-stage airline game is solved numerically. Existence and uniqueness of a second-stage equilibrium are not guaranteed, and the non-linearity of the airline problem implies that a best response iteration may fail to converge. In addition, the solver may converge to different stationary points depending on the initial solution. For these reasons, for each candidate policy, we run the second-stage game from multiple starting points and we repeat the best response cycle using different airline orders. Each airline problem is solved using \textit{IPOPT} \citep{Wachter2006}. For all the scenarios tested, the solutions consistently return the same equilibrium outcome.

Following \citet{Adler2022}, regulators update their first-stage choices using a local grid search around the current incumbent solution. Starting from the incumbent, the algorithm first explores neighboring policies on a discrete grid obtained by moving in absolute increments of \euro~15 in each relevant first-stage dimension. At each candidate grid point, we solve the second-stage problem using the multi-start procedure and evaluate the regulators' objectives. If an improvement is found, the incumbent is updated and the local search is repeated around the new solution. Once the search yields two consecutive iterations with the same incumbent solution (i.e., no change in the selected policy across two rounds), the algorithm switches from fixed-step exploration to a shrinking-grid refinement. The grid radius is reduced by 50\% (relative to the incumbent policy level), and this halving is repeated iteratively to further localize the search. The refinement stops when the cumulative reduction in the grid radius becomes smaller than 2\%, at which point no further meaningful improvements are expected and the current incumbent solution is retained as the final equilibrium outcome.

Finally, while passengers are not modeled as explicit players, their decisions are central to airline objectives in the second stage. Airlines compete for passengers in the routes they serve in order to maximize revenues. Therefore, the market share function endogenises passenger choices over flight alternatives indirectly.

\section{Case Study}\label{sec:case_study}

We assess the model on a representative network spanning two geographic regions, North America and Europe. The network nodes are shown in Figure~\ref{fig:network}. The network comprises 24 nodes, evenly divided across the two regions, and includes domestic, regional, and transatlantic connections. The selected nodes and airlines cover 7\% of demand in each of the three markets.\footnote{A larger set of nodes and airlines would increase computation time substantially given the two-stage nature of the game. Nevertheless, because we use the same hubs as in \cite{adler2026estimating}, the selected nodes adequately capture differences in traffic volume across the two regions.} Average inter-node distances are also consistent with observed values (EU: 1,250 km; NA: 2,076 km; TRA: 6,959 km). Transatlantic connections are operated by legacy carriers through their hubs, whereas low-cost carriers provide point-to-point service within each region. We classify Southwest, Spirit, and JetBlue as North American low-cost carriers and Ryanair, easyJet, and Vueling as their European counterparts.

We consider three regulators: two regional regulators, each with authority over one region, and one global regulator with jurisdiction over both regions. Regional and transatlantic connections are charged by the regulator with jurisdiction over the departure airport. In scenarios with a global regulator, all operations are also subject to its charge. We assume a social cost of carbon of \euro~200, consistent with recent IPCC-based estimates \citep{portner2022climate}. Following \citet{lee2021contribution}, we multiply $CO_2$ emissions by a factor of three to account for non-$CO_2$ effects associated with aviation. Appendix~\ref{app:sensitivity} reports a sensitivity analysis using a lower value for carbon damages to assess how a reduction in the social cost of emissions affects the effectiveness of the carbon tax.

\begin{figure}[!t]
  \centering
  \resizebox{\columnwidth}{!}{\includegraphics{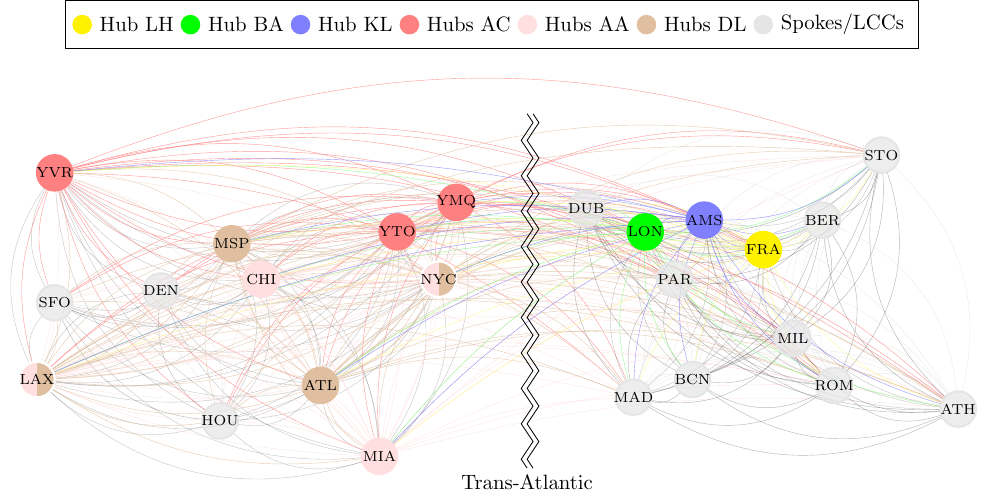}}
  \caption{Selected nodes in North America and Western Europe.}
  \label{fig:network}
\end{figure}

We consider four aircraft variants, two for short-haul flights and two for long-haul flights. Aircraft purchase and salvage values are obtained from manufacturer sources and reported in Table~\ref{tab:aircraft}. We assume a 25\% discount from retail prices. Given the proprietary nature of these discounts and their variance across airlines, the rate is set equal to half the maximum discount reported by Airbus. \footnote{The discount values reported by Airbus are computed by comparing annual report information on aircraft prices and quantities sold with operating revenues.}

\begin{table}[!t]
  \centering
  \caption{Aircraft purchase, salvage values, and operating characteristics.}
  \input{aircraft_costs.tex}
  \label{tab:aircraft}
\end{table}

We begin with the baseline calibration (Section~\ref{baseline}) and then organize the results around five questions. Section~\ref{subsec:airline_responses} examines how airlines respond to environmental charges across regulatory scenarios. Section~\ref{subsec:welfare_comparison} explains how firm-level responses to different institutional arrangements produce varying welfare outcomes. Section~\ref{subsec:co-operation} assesses whether the cooperative equilibrium can be sustained through side payments. Section~\ref{subsec:sensitivity} examines robustness to alternative assumptions about the social cost of carbon.

\subsection{Baseline}
\label{baseline}

We benchmark model outcomes against a baseline scenario without environmental charges, calibrated to approximate 2019 conditions. Passenger demand follows a nested logit model parameterized using the estimates in \citet{adler2026estimating} for the European, North American, and transatlantic markets; the corresponding parameter values are reported in Appendix~\ref{app:logit}.

We use this scenario to assess the model’s ability to reproduce observed outcomes. Table~\ref{tab:validation} reports validation results for one year of operations. Overall, the model matches the main patterns observed for both legacy and low-cost carriers, with the remaining discrepancies reflecting the reduced network representation and data aggregation.

Tables~\ref{tab:welfare_baseline} and \ref{tab:operating_results} show that passenger and airline surpluses are higher in North America than in Europe. In the model, this reflects stronger demand, higher willingness to pay, and longer average regional distances in North America. Regional differences in the attractiveness of the outside option, which captures alternative transport modes and no travel, further widen the gap. Together, these factors generate higher equilibrium fares and larger carrier revenues in North America.

Higher demand in North America leads both legacy and low-cost carriers to operate more flights than in Europe, resulting in higher emissions. Since carbon emissions generate global climate impacts, the model distributes environmental damages equally across both regions, regardless of where the emissions originate. Under zero environmental charges, airlines have little incentive to internalize the emissions externality and fleets consist entirely of older aircraft.

Table~\ref{tab:REGs_results} shows that overall welfare, accounting for both greenhouse gas and non-$CO_2$ emissions, is lowest under the zero-charge scenario because it generates the highest emissions and, hence, the largest externality. This ranking need not hold at the regional level, since consumer and producer surpluses are region specific whereas damages are shared globally by assumption. In North America, welfare is higher under the baseline than under the overlapping-regulator scenario, consistent with the analytics in Section~\ref{sec:analytical_model}, where the baseline can outperform overlapping regulation when $e$ is sufficiently small. Environmental charges nonetheless reduce emissions in every regulated scenario by affecting airfares, frequencies, and fleet choices.

\begin{table}[!t]
  \caption{Operating results in the \textit{baseline} scenario (annual, full network).}
  \resizebox{\textwidth}{!}{\input{real_world}}
  \label{tab:validation}
\end{table}

\begin{table}[!t]
  \caption{Welfare in the \textit{baseline} scenario (annual, full network).}
  \resizebox{\textwidth}{!}{\input{baseline_tab}}
  \label{tab:welfare_baseline}
\end{table}

\begin{table}[!t]
  \caption{Operating results (annual, full network).}
  \resizebox{\textwidth}{!}{\input{Airlines_results}}
  \label{tab:operating_results}
\end{table}

\begin{table}[!t]
  \caption{Results of the case study (annual, full network).}
  \resizebox{\textwidth}{!}{\input{REGs_results}}
  \label{tab:REGs_results}
\end{table}

\begin{table}[!t]
  \label{tab:benchmark_summary}
  \caption{Summary of benchmark comparisons.}
  \resizebox{\textwidth}{!}{
    {\small
      \begin{tabular}{@{}p{0.25\textwidth}p{0.32\textwidth}p{0.36\textwidth}@{}}
        \hline
        Setting & Best-performing regime                                                           & Interpretation \\
        \hline
        Symmetric analytical model
                & Uniform global regulation
                & One common instrument is sufficient when regions are identical.                                   \\[3pt]

        Asymmetric stylized model
                & Decentralized or overlapping regulation can outperform uniform global regulation
                & Multiple instruments can better accommodate local market distortions.                             \\[3pt]

        Calibrated case study
                & Differentiated global regulation
                & Coordination combined with region-specific charges performs best.                                 \\[3pt]

        Participation constraints
                & Decentralized regional regulation may be individually preferred by one region
                & Transfers are needed to sustain efficient coordination.                                           \\
        \hline
      \end{tabular}}
  }
\end{table}

\subsection{Airline strategic responses to environmental charges}\label{subsec:airline_responses}

Airlines do not absorb environmental charges as a uniform cost increase. Their response varies systematically with network structure and business model, and these differences shape both the welfare incidence of regulation and the strategic burden it imposes. Airlines adjust along three jointly determined margins, namely service frequency, fleet deployment and fares. Table~\ref{tab:airline_response_overall} reports the aggregate response relative to the baseline.

\begin{table}[!t]
  \caption{Aggregate airline response by regulatory scenario.}
  \resizebox{\textwidth}{!}{\input{airline_response_summary}}
  {Notes. Frequency, fleet, and passenger columns report percentage changes relative to the baseline. Routes affected is the share of airline-route observations with a nonzero net frequency change.}
  \label{tab:airline_response_overall}
\end{table}

All regulated scenarios reduce frequencies, fleet deployment and passenger volumes while increasing fares. The strongest response occurs under uniform global regulation, followed closely by differentiated global regulation; the weakest occurs under two competing regional regulators, with overlapping regulation in between. Two features stand out. Frequency reductions exceed passenger reductions in every scenario, indicating that airlines contract capacity more aggressively than demand drops in order to preserve aircraft utilization and load factors. Fleet reductions track frequency reductions closely, so the response is a structural contraction in deployed capacity rather than a short-run scheduling adjustment.

\begin{table}[!t]
  \caption{Regional incidence of airline responses.}
  \resizebox{\textwidth}{!}{\input{geographic_incidence}}
  \label{tab:airline_response_region}
\end{table}

Table~\ref{tab:airline_response_region} reports the regional incidence. Since transatlantic markets correspond to long-haul operations and intra-European and intra-North American markets to short-haul operations, the regional split also captures the main haul-type distinction. Transatlantic markets experience the largest proportional frequency reductions, fare increases and passenger losses under every regime, with transatlantic frequencies falling by 16.61\% and weighted fares rising by \euro~40.47 under uniform global regulation. Long-haul markets are therefore the principal channel through which environmental charges affect airline operations and consumers. European and North American intra-regional markets display contrasting adjustment patterns. Intra-European frequencies fall sharply while weighted fares barely change, indicating limited ability to pass policy costs through. North American markets show more moderate frequency reductions and substantially larger fare pass-through, so pass-through ability differs across regional market structures.

\begin{table}[!t]
  \caption{Policy response by airline type.}
  \resizebox{\textwidth}{!}{\input{airlines_business_models}}
  {Notes. Airline groups separate legacy and low-cost carriers by home region. Frequency, fleet, and passenger columns report percentage changes relative to baseline.}
  \label{tab:airline_response_group}
\end{table}

Table~\ref{tab:airline_response_group} shows that responses also vary across airline business models. European low-cost carriers exhibit the largest proportional frequency and fleet reductions, between 15.80\% and 18.18\% across scenarios, while their weighted fare changes remain below \euro~0.25. These carriers behave primarily as capacity adjusters, a response consistent with strong price competition and elastic demand in short-haul point-to-point markets, where raising fares would induce substantial demand losses. Legacy carriers respond differently. American legacy carriers reduce frequency and fleet but also raise fares by \euro~7.96 to \euro~15.39, and European legacy carriers raise fares with smaller passenger losses than other groups. Hub-and-spoke carriers therefore have more scope for fare pass-through, reflecting a passenger mix with higher willingness to pay, greater long-haul exposure and network-based market power.

These results carry three strategic implications. Policy incidence is jointly determined by route geography and business model in ways that average-fare measures obscure. When low-cost carriers respond by contracting service rather than raising fares, consumers are harmed primarily through reduced availability of low-fare options; when legacy carriers pass the burden into fares, consumers experience the policy directly through higher prices. Evaluating environmental policy through average fare changes alone would understate the importance of service availability and misattribute welfare incidence across consumer segments. The case study also shows that regulated scenarios accelerate the replacement of older aircraft with more fuel-efficient versions, so the strategic value of cleaner capital deployment is realized through the interaction of pricing and fleet decisions rather than either margin in isolation. Most importantly, different business models have different preferences over regulatory architecture. A regime that looks similar in aggregate can impose materially different strategic burdens on a regionally focused low-cost carrier than on a transatlantic hub carrier. This heterogeneity in firm preferences over the regulatory environment links the firm-level results in this subsection to the welfare comparisons that follow.

\subsection{Why do regulatory scenarios produce different welfare outcomes?}
\label{subsec:welfare_comparison}

As detailed in Table~\ref{tab:REGs_results},  welfare outcomes vary across regulatory structures because each arrangement creates incentives for firms, which in turn produce different equilibrium levels of emissions, surplus and government revenue. Specifically, (i) decentralized regulation encourages free-riding, (ii) unified global regulation requires balancing market power against environmental priorities and (iii) overlapping policies generate conflicting incentives that distort both regional and global objectives.

(i) When two regional regulators act independently (Scenario~2), they tend to free-ride on each other, setting charges well below the social cost of carbon. Each regulator protects the surplus of both passengers and carriers within its jurisdiction and this motive dominates other strategic incentives in the game, reducing the likelihood of a tax war aimed at extracting surplus from transatlantic passengers in the opposing region. The discrepancy between the European and North American charges is substantial because the European regulator sets a charge of \euro~51 per ton while the North American regulator sets \euro~29. This asymmetry reflects differences in demand elasticity, market structure and the availability of outside transport options. This asymmetry mirrors the institutional landscape that carriers currently face. The European Union prices intra-European emissions through the EU ETS, with
allowances near \euro70 per ton once aviation free allocation is fully phased out
in 2026, whereas North America operates no comparable emissions trading system,
and emissions on international routes are addressed by the much cheaper CORSIA
offsets \footnote{EU ETS allowances traded in a range of approximately \euro60 to \euro95 per ton over 2025 and 2026 \citep{icap2026}, and free allocation to aviation is fully
  phased out from 2026 \citep{ec2026aviation}, so European carriers face close to
  the full allowance price on intra-European flights. CORSIA-eligible unit prices
  were estimated at USD 10 to 40 per ton during the scheme's first phase, based on
  the limited price information available for the ICAO 2025 periodic review
  \citep{easa2025corsia}. These prices are quoted per ton of $CO_2$, whereas the
  social-cost benchmark used here applies an additional factor of three to account
  for non-$CO_2$ forcing.}. The model thus reproduces the observed pattern of a higher
European charge and a lower North American value, with regional charges
remaining well below the social cost of carbon, which suggests that the
market-power and Mohring effect identified in Section~\ref{sec:analytical_model} is relevant. Despite the free-riding, overall welfare increases by \euro~11.21B relative to the baseline. Damages caused by emissions fall by \euro~27.22B and government revenue rises by \euro~20.62B, but these gains come at the expense of producers and consumers, whose surplus declines by \euro~19.47B and \euro~17.14B respectively. This pattern confirms Corollary~\ref{cor:tax_rank} in which local regulators collectively over-tax relative to the global optimum, yet the environmental benefit still outweighs the distortion when the social cost of emissions exceeds the threshold $\bar{e}_{21}$ identified in Proposition~\ref{prop:welfare_pos}.

(ii) Under a uniform global regulator (Scenario~3), free-riding is eliminated and one authority internalizes all emissions across the network. The optimal charge is the highest among all scenarios at \euro~60 per ton, consistent with the analytical results of Section~\ref{sec:analytical_model}. Yet this charge remains well below the social marginal cost of emissions that would compensate for the $CO_2$ and $CO_2$-equivalent emissions produced. The trade-off between accounting for the environmental externality and preserving output in the presence of market power and frequency externality drives the optimal charge below the Pigouvian level (Proposition~\ref{prop:scenario3_SO}). The regulator cannot discriminate across itineraries, so the same charge applies to all operations. This lack of price discrimination may yield a sub-optimal uniform charge. Relative to the baseline, the global charge increases ticket prices and reduces service frequency, while the higher cost also induces airlines to accelerate fleet renewal toward more fuel-efficient aircraft compared to the two regional regulator scenario. Environmentally, the global charge produces the largest reduction in emissions among all scenarios leading to a \euro~40.68B decline in emissions damage, compared with  \euro~27.22B under competing regulators. Overall welfare rises by \euro~13.73B relative to the baseline, confirming the analytical insight that coordination improves welfare by internalizing cross-regional emissions and eliminating regulatory free-riding. However, once regions differ in demand, market structure, and emissions intensity, coordination alone is not sufficient: the instrument must also be regionally differentiated.

Although a uniform global tax improves overall welfare, it is not necessarily optimal when regions are heterogeneous. Regional regulators can better account for differences in passenger preferences, network characteristics and regional climate vulnerability by implementing tailored carbon charges. Under a global regulator that is allowed to set region-specific charges (the discriminating global scenario), the taxation pattern reverses relative to the competing regulator setting. The tax in North America is higher (\euro~60) and broadly aligned with the global level under uniform taxation, reflecting the greater emissions associated with higher traffic volumes. The absence of free-riding or jurisdictional protectionism means the global policy is applied consistently across regions. Europe faces a lower tax rate (\euro~51), reflecting its smaller regional network and lower associated emissions. Although the charges remain below the social cost of emissions, this scenario yields the highest welfare among all scenarios; \euro~13.98B above the baseline. These results suggest that even in the absence of free-riding and carbon leakage, market distortions can lead to inefficient environmental outcomes when regulators rely solely on a uniform carbon charge, consistent with Proposition~\ref{prop:scenario3_SO} and the quantity rankings of Proposition~\ref{prop:rank_pos}.

(iii) The overlapping scenario (Scenario~4) combines two decentralized regional regulators and one global regulator. Accordingly, airlines are subject to double taxation with one charge imposed by the local regulators and another by the global authority. We assume that the global regulator collects the revenues generated from the charge imposed on all flights in the network but cannot retain them, as it operates under a revenue-neutral framework. Consequently, all revenues are redistributed equally to the two local regulators as a lump-sum transfer, and these additional funds are incorporated into their respective objective functions. The equal redistribution of revenues between the two regions may create an incentive for the higher-polluting jurisdiction to free-ride on the other region. North America leverages its higher emissions as a competitive instrument to capture part of the surplus from the cleaner region. The global regulator sets a charge of \euro~53 per ton whilst the local regulators respond by setting zero tax in Europe and a subsidy of \euro~10 per ton in North America. When considering the overall burden on airlines, given the double taxation applied to each flight, the total level is broadly in line with that observed under the competing regulator scenario. This induces similar airline behavior, leading to comparable reductions in emissions across the network.

Despite similar aggregate tax burdens, welfare outcomes differ across regulatory regimes. Relative to the baseline, total emissions decline by \euro~33.84B, producer surplus falls by \euro~24.22B, and consumer surplus falls by \euro~21.74B. Overall welfare nevertheless increases by \euro~13.20B, although this gain remains smaller than under either uniform and differentiated global regulation scenarios. The resulting ordering, $\mathcal{W}^{(3)} > \mathcal{W}^{(4)} > \mathcal{W}^{(2)*}$, is consistent with Section~\ref{subsec:asymmetry}, where we show that layering regional and global regulation does not improve upon global regulation when regional asymmetries are moderate. In this case, the potential advantages of multiple regulatory layers are offset by local free-riding on the global policy, regional protectionism and subsidization, and pre-existing distortions in the airline industry. The analytical model also shows that this ranking need not hold under stronger regional asymmetries, where differentiated regulation may better reflect cross-region heterogeneity.

\subsection{Can co-operation be sustained?}
\label{subsec:co-operation}

Aggregate efficiency is not sufficient for sustainable policy design. Although differentiated global regulation achieves the highest total welfare, the cooperative equilibrium is not participation-compatible for all jurisdictions. As shown in Table~\ref{tab:REGs_results}, the North American regulator experiences a decline in welfare under coordination relative to decentralized regional regulation, driven primarily by reductions in airline profitability and passenger connectivity. This uneven distribution of benefits creates an incentive for North America to defect to a non-cooperative regime, where it can free-ride on European emission reduction efforts while protecting its domestic aviation sector. The result highlights a central political feasibility constraint: the welfare-maximizing sustainability policy may fail unless the jurisdictions that gain from coordination compensate those that bear disproportionate adjustment costs.

To implement the higher joint surplus associated with the globally coordinated outcomes, the jurisdiction with the larger co-operative payoff could make a transfer (side payment) to the jurisdiction with the smaller payoff to satisfy its participation constraint. In our setting, aggregate welfare under co-operation exceeds the sum of the regulators' welfare under competition, so there exists a transfer from Europe to North America that ensures both regulators are weakly better off than in the non-cooperative outcome. In a two-player, transferable-utility setting, this superadditivity condition implies that the core is non-empty \citep{osborne1994course}.

Consider first the uniform global regulation. The co-operative surplus relative to the non-cooperative regime is:
\begin{equation}
  \Delta_c = 13.73 - 11.21 = 2.52
\end{equation}

Let $\psi$ denote a transfer from Europe to North America. To induce North America to join the cooperative arrangement, $\psi$ must satisfy the two participation constraints:
\begin{equation}
  \begin{cases}
    -4.27 + \psi \geq 3.08 \\
    18.01 - \psi \geq 8.13
  \end{cases}
\end{equation}

These imply $\psi \in [7.35, 9.88]$. This range captures the side payment needed for individual rationality; it is distinct from the incremental surplus $\Delta_c$. Any transfer in this interval therefore makes both regulators weakly better off than under non-co-operation. Similarly, when the differentiated global regulator sets region-specific charges, the transfer must lie in $\psi \in [7.84, 10.63]$ to induce North America to cooperate. This side-payment result highlights an equity-efficiency tension in sustainability governance. The regime that maximizes aggregate welfare imposes adjustment costs unevenly across jurisdictions. Without a compensating transfer, the jurisdiction bearing those costs prefers a less efficient decentralized policy. Distributional design is therefore not separate from environmental effectiveness; it is a condition for sustaining efficient climate governance. In this setting, transfers are not merely a fairness correction after the efficient policy has been identified. They are part of the institutional architecture required to make the efficient policy politically feasible.

Consequently, the cooperative outcome achieves aggregate efficiency but fails participation compatibility without transfers. Coordinated governance therefore requires pre-agreed redistribution to remain politically stable. Allowing transfers restores participation compatibility by making cooperation individually rational for both regulators. This finding speaks directly to the institutional design challenges facing ICAO's CORSIA. The absence of transfer mechanisms helps explain why international climate governance often settles on weaker offsetting arrangements rather than binding, coordinated emission pricing. In network industries, therefore, sustainable policy design requires distributional instruments alongside environmental instruments.

\subsection{Sensitivity to the social cost of carbon}
\label{subsec:sensitivity}

The preceding results assume a social cost of carbon $e$ equal to \euro~200 per ton, multiplied by three to also capture non-$CO_2$ effects, in line with the latest IPCC reports \citep{portner2022climate}. Given the substantial uncertainty surrounding this parameter, we conduct a sensitivity analysis that is specified in Appendix~\ref{app:sensitivity}, assuming the social cost is set at 30\% of the baseline value.

The main welfare rankings do not change. The differentiated global regime remains the most efficient outcome and decentralized regional regulators continue to outperform overlapping regulation. However, sensitivity analysis reveals that the likelihood of cooperation becomes even weaker at the lower emission prices. The regulators choose to subsidize airlines in order to mitigate their market power. A global regulator finds it optimal to subsidize the aviation industry, placing greater weight on industry performance and consumer surplus than on environmental concerns but at a large cost to the coordinating governments. Consequently, the incentive for individual regions to abandon the global framework becomes stronger, reinforcing the conclusion that global carbon policy is structurally unstable without binding side-transfers. Unlike the baseline scenario, where Europe compensates North America to sustain cooperation, the lower social cost of carbon reverses the direction of the required transfer. Sustaining the globally efficient policy in this scenario would require North America to transfer an annual side payment of between \euro~8.45 and \euro~10.82 billion to Europe in order to offset the fiscal losses.

\section{Conclusions and future directions}
\label{sec:conclusions}

This paper studies how firms in global network industries plan operations under fragmented environmental regulation, and how the firm-level response shapes the regulatory architectures these firms face. The two-stage game-theoretic framework treats firm pricing, capacity allocation and capital deployment decisions as the primary channel through which environmental regulation produces its effects, and treats the regulatory environment as endogenous to firm behavior. Aviation provides the operational setting because it combines mobile capital, network competition, overlapping environmental regulation and cross-border emissions, but the mechanisms apply more generally to network industries in which firms transmit regulatory signals across jurisdictions through linked operational decisions, including international shipping, electricity transmission, multi-modal freight, logistics networks and global supply chains.

Three findings emerge with direct implications for firm strategy and regulatory design. First, firms in concentrated network markets face a regulatory environment in which the effective charge falls systematically below the social cost of emissions, and may take the form of a subsidy rather than a tax for parameter ranges in which market power and the Mohring effect dominate. Current instruments such as emissions trading schemes and CORSIA implicitly treat the Pigouvian benchmark as the relevant target, but the effective charge should be calibrated to the degree of market concentration and the strength of frequency effects on each route, neither of which current schemes accommodate. For firms, this means that fleet planning and pricing decisions should be made on the basis of effective charges that depend on local market structure rather than on the headline rate.

Second, firms operating under fragmented regulation face an aggregate burden that exceeds what a coordinated regime would produce, even though each individual regulator under-prices relative to the social cost of carbon. Overlapping regulation in which local and global authorities operate simultaneously compounds this distortion rather than correcting for it, particularly when environmental damages are below the thresholds identified in Proposition~\ref{prop:welfare_pos}. Airlines currently face national departure taxes, regional emissions trading schemes and CORSIA on international routes simultaneously, so the implication is that layering additional instruments on top of existing ones is unlikely to reduce the burden firms face without explicit coordination of charge levels across jurisdictions.

Third, the cooperative regulatory regime is not the one firms should plan against. Differentiated global regulation achieves the highest aggregate welfare in the calibrated case study, but it leaves North America worse off than decentralised regional regulation, creating an incentive to defect. Side payments in the range of \euro~7 to \euro~10 billion annually are needed to sustain cooperation, and these mechanisms are absent from existing international aviation climate agreements. This finding reframes equity from a secondary distributional concern into a design condition. In fragmented network industries, the policy that maximises aggregate welfare shifts costs toward particular regions, firms or passenger segments, and if those actors can block, weaken or exit the cooperative arrangement, unequal burden-sharing directly undermines environmental effectiveness. For firms, the strategic implication is that coordination repeatedly settles on weaker arrangements than the welfare-maximising benchmark would suggest, and capacity, fleet and pricing decisions should be planned on the assumption that fragmented regulation is the durable equilibrium.

These findings carry three specific managerial implications. Business fares absorb a disproportionate share of regulatory costs because business demand is less elastic, and this effect varies across jurisdictions, so revenue management systems should model regulatory charge pass-through by demand segment and jurisdiction jointly rather than by route alone. Welfare rankings across regulatory regimes reverse as regional cost structures diverge, so the regulatory environment is structurally unstable in direction rather than merely uncertain in magnitude, and carriers whose capacity and fleet commitments span multiple jurisdictions are more exposed to these reversals than carriers operating within a single region. Different business models also have different regulatory interests. A low-cost carrier operating within a single region would benefit from coordination that eliminates the excess taxation produced by fragmented regulation, whereas a hub-and-spoke carrier connecting asymmetric regions may be better served by decentralised or overlapping regulation that accommodates the heterogeneity it faces across its network.

The analysis is subject to several limitations that suggest directions for future research. The set of airports served by an airline is fixed, so airlines cannot enter new routes in response to charges. Relaxing this assumption would enable the study of carbon leakage through network reconfiguration. Environmental damage is shared equally across jurisdictions, which abstracts from heterogeneous climate vulnerability. The parameter $\eta_r$ in the extended model provides a starting point for exploring this dimension. Introducing profit-maximising airports as a third stage in the game would capture the interaction between landing charges, environmental taxes and airline network decisions, and would also allow the framework to incorporate charges for local externalities such as air pollutants and aircraft noise that are relevant to airport-level environmental regulation.

\section*{Acknowledgments}
The authors would like to thank the participants at INFORMS 2022, hEART 2023, ITEA 2023, ISMP 2024, GARS, AEC, and the SOAR 2023 and SOAR 2026 conferences for their helpful comments. Nicole Adler and Gianmarco Andreana gratefully acknowledge support from the Israel Science Foundation under grant 2441/21, which helped finance this research. Nicole Adler also gratefully acknowledges partial funding from the Goldman Center for Data-Driven Innovation at the Hebrew University Business School. Gerben de Jong gratefully acknowledges partial funding from the Netherlands Organisation for Scientific Research (NWO) through the Rubicon programme, award 019.201SG.019.

\startappendices
\renewcommand{\thelemma}{\thesection.\arabic{lemma}} 

\section{Nested logit demand parameters}
\label{app:logit}

This appendix reports the nested logit demand parameters used in the case study. The coefficients are taken from the demand estimation in \citet{adler2026estimating} and enter the systematic utility specification in Equation~\ref{eq:sys_utility}. Table~\ref{tab:logit_coefficients} reports the parameter values separately by passenger segment and market, which are used to compute market shares through the nested logit formulation in Equation~\ref{eq:marketShare}.

\begin{table}[!ht]
  \centering
  \caption{Nested logit coefficients from \citet{adler2026estimating}.}
  \input{logit.tex}
  \label{tab:logit_coefficients}
\end{table}

\FloatBarrier

\section{Sensitivity analysis over the social cost of carbon}
\label{app:sensitivity}

Given the uncertainty surrounding the monetary valuation of carbon emissions, in this Appendix, we present a sensitivity analysis of carbon charges and social welfare assuming that the social cost of $CO_2$, $e$, is set at 30\% of the baseline value. Results for all scenarios are reported in Tables~\ref{tab:baseline_sens_30_tab} and ~\ref{tab:REGs_sensitivity_30}.

When the social cost of emissions is lower, regulators choose to subsidize airlines in order to mitigate market power. The most efficient outcome continues to be the differentiated global regime. A global regulator finds it optimal to subsidize the aviation industry, placing greater weight on industry performance and output than on environmental concerns. As a result, governments allocate taxpayer resources to the sector. Airlines, benefiting from the subsidy, expand their operations, leading to higher passenger volumes and increased revenues. This effect is pronounced under the uniform and differentiated global regulation scenarios. Consumers benefit concurrently from improved connectivity and from the reduction in airline market power induced by the subsidy.

In the case of overlapping regulators, the overall fiscal balance remains negative. However, the European regulator sets a positive carbon tax (+\euro~66/ton) to counteract the global subsidy (-\euro~83/ton) and partially offset the resulting increase in emissions. This reflects the relatively smaller contribution of the European aviation sector to regional welfare compared to North America, suggesting that the European regulator will collect tax revenue rather than supplementing the global policy.

As in the case where the marginal cost of damage from emissions is higher, the sustainability of a cooperative carbon policy is not guaranteed. Even when the social cost is set at 30\% of the baseline value, cooperation remains fragile. For relatively low levels of the social cost of carbon, Europe has an incentive to deviate. While the cooperative subsidies generate surplus gains for European passengers and airlines, these benefits are offset by the fiscal burden imposed on the government, with a deficit of \euro~32.25 billion under the differentiated uniform global regulation scenario. Consequently, Europe's overall welfare is lower under cooperation than under decentralized competition, -\euro~13.81 billion compared to -\euro~5.36 billion, giving the European regulator an incentive to abandon the global framework in order to protect its fiscal balance.

To implement the higher joint surplus associated with the globally coordinated outcomes in this scenario, North America must provide a transfer to Europe to satisfy its participation constraint. The cooperative surplus relative to the non-cooperative regime remains strictly positive, allowing for a mutually beneficial transfer. Let $\psi$ denote a transfer from North America to Europe. To induce Europe to maintain the cooperative arrangement under a differentiated global regulator, $\psi$ must satisfy the participation constraints for both regions:
\begin{equation}
  \begin{cases}
    -13.81 + \psi \geq -5.36 \\
    24.23 - \psi \geq 13.41
  \end{cases}
\end{equation}

These inequalities imply that the stabilizing transfer must fall within the range $\psi \in [8.45, 10.82]$ in billions of euros. Any transfer within this interval compensates Europe for its fiscal losses while allowing North America to retain a portion of its cooperative gains. Therefore, analogous to the baseline analysis, allowing for inter-regional side payments restores efficiency by making cooperation individually rational for both regulators.

\begin{table}[!t]
  \centering
  \caption{Welfare in the baseline scenario when the social cost of carbon is 30\% of $e$ annually.}
  \input{baseline_sens_30}
  \label{tab:baseline_sens_30_tab}
\end{table}

\begin{table}[!t]
  \centering
  \caption{Sensitivity analysis when the social cost of carbon is 30\% of $e$ annually.}
  \input{sensitivity_30}
  \label{tab:REGs_sensitivity_30}
\end{table}

\FloatBarrier

\section{Proofs of Lemmas and Propositions}

\subsection{\texorpdfstring{Auxiliary derivations: $\phi$ and consumer surplus}{Auxiliary derivations: phi and consumer surplus}}\label{app:aux}

\begin{lemma}[Derivatives and bounds for $\phi$]\label{lem:phi_app}
  Let $\phi(q) = \ln(1 + q/s) + \frac{q}{s + q}$ with $s > 0$. Then
  \[
    \phi'(q) = \frac{2s + q}{(s + q)^2} > 0,
    \qquad
    \phi''(q) = - \frac{3s + q}{(s + q)^3} < 0.
  \]
  Moreover $\phi'(q) \leq \phi'(0) = 2 / s$ for all $q \geq 0$.
\end{lemma}
\begin{proof}
  Differentiate term by term:
  \begin{equation*}
    \frac{d}{dq}\ln\Bigl(1+\frac{q}{s}\Bigr)  = \frac{1}{1+q/s}\cdot\frac{1}{s} = \frac{1}{s+q},
    \qquad
    \frac{d}{dq}\Bigl(\frac{q}{s+q}\Bigr)
    = \frac{(s+q)-q}{(s+q)^2}
    = \frac{s}{(s+q)^2}.
  \end{equation*}
  Summing yields
  \[
    \phi'(q) = \frac{1}{s + q} + \frac{s}{(s + q)^2} = \frac{s + q + s}{(s + q)^2} = \frac{2s + q}{(s + q)^2}.
  \]
  Differentiate $\phi'(q) = (2s + q)(s + q)^{-2}$:
  \begin{equation*}
    \phi''(q) = (s+q)^{-2} + (2s+q)(-2)(s+q)^{-3} = -\frac{3s+q}{(s+q)^3}.
  \end{equation*}
  Since $\phi''(q) < 0$, $\phi'$ is decreasing in $q$, hence $\phi'(q) \leq \phi'(0) = 2/s$.
\end{proof}

\begin{lemma}[Consumer surplus]\label{lem:CS_app}
  With utility \eqref{eq:utility} and inverse demand \eqref{eq:inverse_demand}, consumer surplus is
  \begin{equation}
    CS(q_1, q_2) = \frac{1}{2}(q_1 + q_2)^2
    - \beta \left(\frac{q_1^2}{s + q_1} + \frac{q_2^2}{s + q_2}\right).
    \label{eq:CS_app}
  \end{equation}
  Under symmetry $q_1 = q_2 = q$, $CS(2q) = 2q^2 - \frac{2 \beta q^2}{s + q}$ and $\frac{1}{2} CS(2q) = q^2 - \frac{\beta q^2}{s + q}$.
\end{lemma}

\begin{proof}
  By definition, $CS = U - \sum_{i=1}^2 p_i q_i$. Let $Q:= q_1 + q_2$. From \eqref{eq:inverse_demand}:
  \[
    p_i q_i = \Big(\alpha + \beta \ln(1 + q_i/s)\Big) q_i + \beta \frac{q_i^2}{s + q_i} - Q q_i.
  \]
  Summing over $i$ gives:
  \begin{equation*}
    \sum_{i=1}^2 p_i q_i = \sum_{i=1}^2 \Bigl(\alpha
    + \beta\ln(1 + q_i/s)\Bigr) q_i + \beta\left(\frac{q_1^2}{s+q_1}
    + \frac{q_2^2}{s+q_2}\right) - Q^2.
  \end{equation*}
  Subtracting from $U = \sum_i(\alpha + \beta \ln(1 + q_i/s)) q_i - \frac{1}{2} Q^2$ yields:
  \begin{equation*}
    CS = \tfrac{1}{2}Q^2 - \beta\left(\frac{q_1^2}{s+q_1} + \frac{q_2^2}{s+q_2}\right),
  \end{equation*}
  which is \eqref{eq:CS_app}. Symmetric expressions follow by substitution.
\end{proof}

\begin{proof}[Proof of Lemma~\ref{lem:cournot}]
  \label{app:cournot}

  Fix $T$. Airline $i$'s profit is \eqref{eq:profit_analytics} with $p_i(q_i, q_j) = \alpha + \beta \phi(q_i) - (q_i + q_j)$ and $m(T) = c_p + \frac{c_f + T}{s}$. Differentiating with respect to $q_i$:
  \[
    \frac{\partial \pi_i}{\partial q_i} = \big(p_i(q_i, q_j) - m(T) \big) + q_i \frac{\partial p_i}{\partial q_i}.
  \]
  Since $\frac{\partial p_i}{\partial q_i} = \beta \phi'(q_i) - 1$, the first-order condition is $p_i - m(T) + q_i(\beta \phi'(q_i) - 1) = 0$. Under symmetry ($q_i = q_j = q$), this simplifies to the equilibrium condition \eqref{eq:Gsym}:
  \[
    G(q; T) := \alpha - m(T) + \beta \phi(q) + \beta q \phi'(q) - 3q = 0.
  \]

  \smallskip
  \noindent To determine the effect of the instrument $T$, we implicitly differentiate the equilibrium condition $G(q^*; T) = 0$ with respect to $T$:
  \[
    G_q(q^*; T) \frac{\partial q^*}{\partial T} + G_T(q^*; T) = 0.
  \]
  The partial derivatives are calculated as:
  \begin{equation*}
    G_T(q; T) = -\frac{1}{s},
    \qquad
    G_q(q; T) = 2\beta\phi'(q) + \beta q\phi''(q) - 3.
  \end{equation*}
  Solving for the change in quantity:
  \begin{equation}
    \frac{\partial q^*}{\partial T} = - \frac{G_T}{G_q} = \frac{1/s}{2 \beta \phi'(q^*) + \beta q^*\phi''(q^*) - 3}.
    \label{eq:dq_formula_app}
  \end{equation}

  Under the standard stability condition
  \begin{equation}
    2\beta\phi'(q^*) + \beta q^*\phi''(q^*) < 1,
    \label{eq:stability}
  \end{equation}
  the denominator in \eqref{eq:dq_formula_app} is less than $-2$. Thus, $\partial q^* / \partial T < 0$, meaning the equilibrium quantity is strictly decreasing in the policy instrument $T$.

  \smallskip
  \noindent For the equilibrium price $p^*(T) = \alpha + \beta \phi(q^*) - 2 q^*$, we find:
  \[
    \frac{\partial p^*}{\partial T} = \big(\beta \phi'(q^*) - 2\big) \frac{\partial q^*}{\partial T}.
  \]
  By Lemma~\ref{lem:phi_app}, $\phi'(q^*) \leq 2/s$. Given the assumption $\beta < s/2$, it follows that $\beta \phi'(q^*) \leq 2 \beta/s < 1$. Consequently, $(\beta \phi'(q^*) - 2) < 0$. Since $\partial q^* / \partial T < 0$, the product is positive: $\partial p^* / \partial T > 0$.
\end{proof}

\begin{proof}[Proof of Proposition~\ref{prop:best_responses}]
  \label{app:best_responses}

  We prove parts (i) and (ii) in turn.

  \smallskip
  \noindent\textbf{(i) Local regulator.}
  Consider local regulator $L$ in Scenario~2. Regulator $L$ values (i) consumer surplus of its citizens, assumed to be half the market, (ii) the locally based airline's profit, (iii) its own tax revenue, and (iv) half the environmental damage. Hence, its objective is
  \begin{equation}
    SW_L = \frac{1}{2} CS(q_1, q_2) + \pi_1 + \tau_L \frac{q_1 + q_2}{s} - \frac{e}{2} \frac{q_1 + q_2}{s}.
    \label{eq:SWL_app}
  \end{equation}
  Under symmetry $q_1 = q_2 = q$, we have $Q = 2q$ and (Lemma~\ref{lem:CS_app})
  \[
    \frac{1}{2} CS(2q) = q^2 - \frac{\beta q^2}{s + q}.
  \]
  Firm 1's profit is $\pi_1 = \big[p(q,q) - m(T) \big] q - F$ with $T = \tau_L+\tau_{L'}$. Using the symmetric Cournot FOC $G(q; T) = 0$ (Lemma~\ref{lem:cournot}), we can write the markup identity
  \begin{equation}
    p(q, q) - m(T) = q \big[1 - \beta \phi'(q) \big],
    \label{eq:markup_app}
  \end{equation}
  so $\pi_1 = q^2 \big[1 - \beta \phi'(q) \big] - F$. Substituting these expressions into \eqref{eq:SWL_app} yields
  \begin{equation}
    SW_L(q, \tau_L) = \Bigl(q^2 - \frac{\beta q^2}{s+q}\Bigr) + q^2\bigl(1 - \beta\phi'(q)\bigr) + \frac{2q}{s}\tau_L - \frac{e}{s}q - F.
    \label{eq:SWL_sym_app}
  \end{equation}

  Regulator $L$ chooses $\tau_L$, taking $\tau_{L'}$ as given, anticipating that $q$ adjusts according to the symmetric Cournot equilibrium. To compute the optimal $\tau_L$ at a candidate symmetric output $q$, we differentiate $SW_L$ along the equilibrium path. The equilibrium path is characterized by the symmetric Cournot condition \eqref{eq:Gsym}:
  \[
    \alpha - \Big(c_p + \frac{c_f + \tau_L + \tau_{L'}}{s} \Big) + \beta\phi(q) + \beta q \phi'(q) - 3q = 0.
  \]
  Solving for $\tau_L$ gives
  \begin{equation}
    \tau_L(q) = s \Big(\alpha - c_p + \beta \phi(q) + \beta q \phi'(q) - 3q \Big) - c_f - \tau_{L'}.
    \label{eq:tauL_path_app}
  \end{equation}
  Since $\tau_{L'}$ is fixed from regulator $L$'s viewpoint, differentiating \eqref{eq:tauL_path_app} yields
  \begin{equation}
    \frac{d \tau_L}{d q} = s \Big(2 \beta \phi'(q) + \beta q \phi''(q) - 3\Big).
    \label{eq:dtauL_dq_path_app}
  \end{equation}

  \noindent Now differentiate \eqref{eq:SWL_sym_app} with respect to $q$ and apply the chain rule:
  \[
    \frac{d SW_L}{d q} = \left.\frac{\partial SW_L}{\partial q}\right|_{\tau_L} + \frac{\partial SW_L}{\partial \tau_L} \frac{d \tau_L}{d q}.
  \]
  From \eqref{eq:SWL_sym_app}, $\frac{\partial SW_L}{\partial \tau_L} = \frac{2q}{s}$, and
  \[
    \frac{\partial}{\partial q} \Big(q^2 - \frac{\beta q^2}{s + q}\Big) = 2q - \beta \frac{2sq + q^2}{(s + q)^2},
  \]
  \[
    \frac{\partial}{\partial q} \Big(q^2(1 - \beta\phi'(q))\Big) = 2q (1 - \beta\phi'(q)) - \beta q^2 \phi''(q),
  \]
  and $\frac{\partial}{\partial q} \big(\frac{2q}{s} \tau_L\big) = \frac{2}{s} \tau_L$ holding $\tau_L$ fixed.
  Therefore
  \begin{equation}
    \frac{d SW_L}{d q} = \Bigl[2q - \beta\frac{2sq+q^2}{(s+q)^2}\Bigr] + \Bigl[2q\bigl(1 - \beta\phi'(q)\bigr) - \beta q^2\phi''(q)\Bigr] + \frac{2}{s}\tau_L + \frac{2q}{s}\frac{d\tau_L}{d q} - \frac{e}{s}.
    \label{eq:dSWL_dq_raw_app}
  \end{equation}
  Substitute \eqref{eq:dtauL_dq_path_app} into \eqref{eq:dSWL_dq_raw_app}:
  \begin{equation*}
    \frac{d SW_L}{d q} = \Bigl[2q - \beta\frac{2sq+q^2}{(s+q)^2}\Bigr] + \Bigl[2q\bigl(1 - \beta\phi'(q)\bigr) - \beta q^2\phi''(q)\Bigr] + \frac{2}{s}\tau_L + 2q\Bigl[2\beta\phi'(q) + \beta q\phi''(q) - 3\Bigr] - \frac{e}{s}.
  \end{equation*}
  Collect terms in $q$ and use Lemma~\ref{lem:phi_app}. A direct simplification gives:
  \begin{equation}
    \frac{d SW_L}{d q} = \frac{2}{s} \tau_L - \frac{e}{s} - 2q + \frac{2 \beta s^2 q}{(s + q)^3}.
    \label{eq:dSWL_simplifiee_app}
  \end{equation}

  \noindent We note that the combination of $\phi'$ and $\phi''$ terms collapses to $\frac{2 \beta s^2 q}{(s + q)^3}$ after substituting $\phi'(q) = \frac{2s + q}{(s + q)^2}$ and $\phi''(q) = -\frac{3s + q}{(s + q)^3}$ and collecting over $(s + q)^3$.

  The local regulator's optimality condition along the equilibrium path is $\frac{dSW_L}{dq}=0$.
  Solving \eqref{eq:dSWL_simplifiee_app} for $\tau_L$ yields
  \[
    \tau_L^*(q) = \frac{1}{2} \left[e + 2sq - \frac{2 \beta s^3 q}{(s + q)^3} \right],
  \]
  which is \eqref{eq:tauL}. Symmetry implies $\tau_{L'}^*(q) = \tau_L^*(q)$.

  \smallskip
  \noindent\textbf{(ii) Global regulator.}
  Global (network) welfare is \eqref{eq:W_def}. Since tax payments are transfers, $T$ affects $\mathcal W$ only through the induced equilibrium quantity $q$. Under symmetry, $\mathcal W(q; e)$ is given by \eqref{eq:W_sym}.

  In Scenario~3, the global regulator chooses $\tau_G$ (hence $T = \tau_G$) anticipating the induced equilibrium quantity $q^*(\tau_G)$ from Lemma~\ref{lem:cournot}. The regulator solves:
  \[
    \max_{\tau_G} \mathcal W \big(q^*(\tau_G); e \big).
  \]
  Under uniqueness and Lemma~\ref{lem:cournot}, $\frac{\partial q^*}{\partial \tau_G}\neq 0$, so the first-order condition is:
  \begin{equation}
    \mathcal W'(q^*; e) = 0.
    \label{eq:global_FOC_app}
  \end{equation}
  By Lemma~\ref{lem:welfare_concavity}, \eqref{eq:global_FOC_app} pins down the unique first-best quantity $q^{SO}(e)$, so the global regulator implements $q = q^{SO}(e)$.

  To compute the corresponding charge as a function of $q$, combine the first-best condition (from \eqref{eq:qSO}):
  \begin{equation}
    c_p + \frac{c_f + e}{s} = \alpha + \beta \phi(q) - 2q
    \label{eq:first_best_condition_app}
  \end{equation}
  with the symmetric Cournot condition \eqref{eq:Gsym}:
  \begin{equation}
    c_p + \frac{c_f + \tau_G}{s} = \alpha + \beta \phi(q) + \beta q \phi'(q) - 3q.
    \label{eq:cournot_condition_app}
  \end{equation}
  Subtract \eqref{eq:first_best_condition_app} from \eqref{eq:cournot_condition_app} to obtain:
  \[
    \frac{\tau_G - e}{s} = \beta q \phi'(q) - q,
  \]
  hence:
  \[
    \tau_G^*(q) = e - sq \big[1 - \beta \phi'(q) \big],
  \]
  which is \eqref{eq:tauG}.
\end{proof}

\begin{proof}[Proof of Corollary~\ref{cor:tax_rank}]
  \label{app:tax_rank}

  Using \eqref{eq:tauL} and \eqref{eq:tauG}:
  \begin{equation*}
    2\tau_L^*(q) - \tau_G^*(q) = sq\left[3 - \beta\left(\frac{2s^2}{(s+q)^3} + \frac{2s+q}{(s+q)^2}\right)\right].
  \end{equation*}
  For any $q \geq 0$, the two terms in parentheses satisfy:
  \[
    \frac{2s^2}{(s + q)^3} \leq \frac{2s^2}{s^3} = \frac{2}{s},
    \qquad
    \frac{2s + q}{(s + q)^2} \leq \frac{2s}{s^2} = \frac{2}{s},
  \]
  so their sum is at most $4/s$. Hence under $\beta < s/2$:
  \[
    \beta \left(\frac{2s^2}{(s + q)^3} + \frac{2s + q}{(s + q)^2}\right) \leq \beta \frac{4}{s} < 2,
  \]
  and therefore the bracketed term is strictly larger than $3 - 2 = 1$. Since $sq > 0$, we obtain $2 \tau_L^*(q) - \tau_G^*(q) > 0$.

  Consequently, since $e \geq 0$, $s > 0$, $q > 0$, and $\beta > 0$:
  \begin{equation*}
    T^{(4)}(q) = 2 \tau_L^*(q) + \tau_G^*(q) = 2e + sq + \frac{\beta s q(3sq + q^2)}{(s + q)^3} > 0.
  \end{equation*}
\end{proof}

\begin{proof}[Proof of Proposition~\ref{prop:rank_pos}]
  \label{app:rank_pos}

  It is useful to rewrite the symmetric Cournot condition \eqref{eq:Gsym} as a mapping from quantities to instruments. Rearranging \eqref{eq:Gsym} yields:
  \begin{equation}
    T = H(q) := s \Big[\alpha - c_p + \beta \phi(q) + \beta q \phi'(q) - 3q \Big] - c_f.
    \label{eq:H_app}
  \end{equation}
  Thus, for any instrument $T$, the induced symmetric Cournot quantity $q$ satisfies $T = H(q)$. Differentiate:
  \begin{equation}
    H'(q) = s \big[2 \beta \phi'(q) + \beta q \phi''(q) - 3\big].
    \label{eq:H_prime_app}
  \end{equation}
  Under the stability condition \eqref{eq:stability} at the relevant equilibria, $2 \beta \phi'(q) + \beta q \phi''(q) < 1$, hence:
  \begin{equation}
    H'(q) < s(1 - 3) = -2s < 0,
    \label{eq:H_decreasing_app}
  \end{equation}
  so $H$ is strictly decreasing.

  Each scenario $\sigma$ pins down an equilibrium $q^{(\sigma)*}$ via a fixed point:
  \begin{align*}
    T^{(\sigma)*} & = H\bigl(q^{(\sigma)*}\bigr), \\
    T^{(1)}       & = 0,                          \\
    T^{(2)}(q)    & = 2\tau_L^*(q),               \\
    T^{(3)}(q)    & = \tau_G^*(q),                \\
    T^{(4)}(q)    & = 2\tau_L^*(q) + \tau_G^*(q).
  \end{align*}
  We proceed in three steps.

  \smallskip
  \noindent\textbf{Step 1: $q^{(2)*} < q^{(3)*}$.} At $q^{(3)*}$, by definition $H(q^{(3)*}) = \tau_G^*(q^{(3)*})$. Then by Corollary ~\ref{cor:tax_rank}:
  \[
    H(q^{(3)*}) - 2 \tau_L^*(q^{(3)*}) = \tau_G^*(q^{(3)*}) - 2 \tau_L^*(q^{(3)*}) < 0
  \]
  Define $F_2(q) := H(q) - 2 \tau_L^*(q)$. Since $H$ is strictly decreasing and $2\tau_L^*(q)$ is continuous, $F_2$ crosses zero at $q^{(2)*}$. Moreover, $H'(q) < 0$ and $2 \tau_L^*(q)$ is weakly increasing in $q$ under $\beta < s/2$, hence $F_2$ is strictly decreasing. Since $F_2(q^{(3)*}) < 0$ and $F_2(q^{(2)*}) = 0$, it follows that $q^{(2)*} < q^{(3)*}$.

  \smallskip
  \noindent\textbf{Step 2: $q^{(4)*} < q^{(2)*}$.} Under the assumption $\tau_G^*(q^{(3)*}) \geq 0$ and Lemma~\ref{lem:phi_app}, $\tau_G^*(q)$ is decreasing in $q$ (since $\frac{d \tau_G^*}{dq} = -s + \frac{2\beta s^3}{(s + q)^3} < 0$ under $\beta < s/2$). Thus $q^{(2)*} < q^{(3)*}$ implies $\tau_G^*(q^{(2)*}) \geq \tau_G^*(q^{(3)*}) \geq 0$, hence:
  \begin{equation*}
    T^{(4)}(q^{(2)*}) = 2\tau_L^*(q^{(2)*}) + \tau_G^*(q^{(2)*}) > 2\tau_L^*(q^{(2)*}) = H(q^{(2)*}).
  \end{equation*}
  Equivalently, $H(q^{(2)*}) - T^{(4)}(q^{(2)*}) < 0$. Define $F_4(q) := H(q) - T^{(4)}(q)$. As above, $H$ is strictly decreasing, $T^{(4)}(q)$ is continuous and $T^{(4)}(q)$ is weakly increasing in $q$ under $\beta < s/2$, so $F_4$ is strictly decreasing. Since $F_4(q^{(4)*}) = 0$ and $F_4(q^{(2)*}) < 0$, we conclude $q^{(4)*} < q^{(2)*}$.

  \smallskip
  \noindent\textbf{Step 3: $q^{(3)*} \leq q^{(1)*}$.} In Scenario~1, $T^{(1)} = 0$, so $q^{(1)*}$ satisfies $H(q^{(1)*}) = 0$. In Scenario~3, $T^{(3)*} = \tau_G^*(q^{(3)*}) \geq 0$, so $H(q^{(3)*}) \geq 0$. Since $H$ is strictly decreasing, $H(q^{(3)*}) \geq H(q^{(1)*}) = 0$ implies $q^{(3)*} \leq q^{(1)*}$, with strict inequality if $\tau_G^*(q^{(3)*}) > 0$.

  Finally, because $H$ is strictly decreasing, the quantity ordering implies the instrument ordering: if $q^{(4)*} < q^{(2)*} < q^{(3)*} \leq q^{(1)*}$ then
  \begin{equation*}
    T^{(4)*} = H(q^{(4)*}) > H(q^{(2)*}) = T^{(2)*} > H(q^{(3)*}) = T^{(3)*}  \geq H(q^{(1)*}) = 0.
  \end{equation*}
  This gives \eqref{eq:T_rank_pos} and \eqref{eq:q_rank_pos}. Price ordering \eqref{eq:p_rank_pos} follows from Lemma~\ref{lem:cournot}, which implies that $p^*(T)$ is increasing in $T$.
\end{proof}

\subsection{Proof of Lemma~\ref{lem:welfare_concavity} and Proposition~\ref{prop:scenario3_SO}}
\label{app:welfare_core}

\begin{proof}[Proof of Lemma~\ref{lem:welfare_concavity}.]

  Under symmetry, \eqref{eq:W_sym} holds. Differentiate:
  \[
    \frac{d}{d q} \left[q \ln \Big(1 + \frac{q}{s} \Big) \right] = \ln \Big(1 + \frac{q}{s} \Big) + \frac{q}{s + q} = \phi(q),
  \]
  hence:
  \[
    \mathcal W'(q; e) = 2 \left[\alpha - c_p - \frac{c_f + e}{s} + \beta \phi(q) - 2q\right],
  \]
  which is \eqref{eq:qSO} when set to zero.
  Differentiate again:
  \[
    \mathcal W''(q; e) = 2 \beta \phi'(q) - 4.
  \]
  By Lemma~\ref{lem:phi_app}, $\phi'(q) \leq 2/s$, so:
  \[
    \mathcal W''(q; e) \leq 2 \beta \frac{2}{s} - 4 = \frac{4\beta}{s} - 4 < 0
    \quad
    \text{under } \beta < s/2.
  \]
  Thus $\mathcal W(\cdot; e)$ is strictly concave and has a unique maximum $q^{SO}(e)$ solving \eqref{eq:qSO}.
\end{proof}

\begin{proof}[Proof of Proposition~\ref{prop:scenario3_SO}]

  In Scenario~3, the global regulator chooses $\tau_G$ anticipating the induced Cournot equilibrium $q^*(\tau_G)$. Because tax revenue is a transfer, global welfare depends on $\tau_G$ only through $q$, so the regulator solves $\max_{\tau_G} \mathcal W(q^*(\tau_G); e)$. Under uniqueness, $\partial q^*/ \partial \tau_G \neq 0$, so the first-order condition is $\mathcal W'(q^*; e) = 0$. By Lemma~\ref{lem:welfare_concavity}, this yields $q^* = q^{SO}(e)$.

  To compute the implementing tax, combine the first-best condition \eqref{eq:qSO} with the symmetric Cournot condition \eqref{eq:Gsym} evaluated at $q^{SO}(e)$, exactly as in \eqref{eq:first_best_condition_app}--\eqref{eq:cournot_condition_app} in Appendix~\ref{app:best_responses}. This gives:
  \[
    \tau_G^{(3)*}(e)
    = e - s q^{SO}(e) \Big[1 - \beta \phi'(q^{SO}(e)) \Big],
  \]
  which is \eqref{eq:tauG_SO}. Since $\mathcal W(\cdot; e)$ is strictly concave, $q^{SO}(e)$ maximizes global welfare, so Scenario~3 weakly dominates any other regime that induces $q \neq q^{SO}(e)$.
\end{proof}

\begin{proof}[Proof of Proposition~\ref{prop:welfare_pos}]
  \label{app:welfare_pos}

  Assume $\tau_G^{(3)*}(e) > 0$. By Proposition~\ref{prop:scenario3_SO}, $q^{(3)*}(e) = q^{SO}(e)$ uniquely maximizes $\mathcal W(\cdot; e)$, so $\mathcal W^{(3)*}(e) > \mathcal W^{(\sigma)*}(e)$ for $\sigma \in \{1,2,4\}$ whenever $q^{(\sigma)*} \neq q^{(3)*}$.

  Under $\tau_G^{(3)*}(e) > 0$, Proposition~\ref{prop:rank_pos} gives $q^{(4)*} < q^{(2)*} < q^{(3)*} < q^{(1)*}$. By strict concavity (Lemma~\ref{lem:welfare_concavity}), $\mathcal W(\cdot; e)$ is strictly increasing on $[0, q^{(3)*}]$ and strictly decreasing on $[q^{(3)*}, \infty)$. Therefore, since $q^{(4)*} < q^{(2)*} < q^{(3)*}$,
  \[
    \mathcal W^{(2)*}(e) = \mathcal W(q^{(2)*}; e) > \mathcal W(q^{(4)*}; e) = \mathcal W^{(4)*}(e),
  \]
  and since $q^{(1)*}>q^{(3)*}$,
  \[
    \mathcal W^{(3)*}(e) = \mathcal W(q^{(3)*}; e) > \mathcal W(q^{(1)*}; e) = \mathcal W^{(1)*}(e).
  \]
  This proves \eqref{eq:W_rank_pos}.
\end{proof}

\smallskip
\noindent\textbf{Threshold ordering.} We have shown $\mathcal W^{(2)*}(e) > \mathcal W^{(4)*}(e)$ for all $e$ satisfying $\tau_G^{(3)*}(e) > 0$. Hence, for any $e$ such that $\mathcal W^{(4)*}(e) \geq \mathcal W^{(1)*}(e)$, it follows that $\mathcal W^{(2)*}(e) > \mathcal W^{(4)*}(e) \geq \mathcal W^{(1)*}(e)$. Thus $\{e: \mathcal W^{(4)*}(e) \geq \mathcal W^{(1)*}(e)\} \subseteq \{e:\mathcal W^{(2)*}(e) \geq \mathcal W^{(1)*}(e)\}$, implying $\bar e_{21} \leq \bar e_{41}$.

\smallskip
\noindent\textbf{Three-case ranking.} By Lemma~\ref{lem:welfare_crossing}, the thresholds $\bar{e}_{21}$ and $\bar{e}_{41}$ exist and are unique. The three cases follow directly from the definitions of the thresholds and the ordering $\bar{e}_{21} \leq \bar{e}_{41}$.

\subsection{Welfare Thresholds}
\label{app:thresholds}

This subsection characterizes the welfare thresholds $\bar e_{21}$ and $\bar e_{41}$ used in Proposition~\ref{prop:welfare_pos}.
From \eqref{eq:Gsym}, the symmetric equilibrium quantity $q^*(T)$ satisfies \eqref{eq:H_app}.

\begin{lemma}[$e$ and $q$ in Scenarios 2 and 4]\label{lem:dq_maps_24_app}
  Let $q$ be a symmetric equilibrium quantity.

  \smallskip
  \noindent\textbf{Scenario 2.}
  At symmetry, $T^{(2)}(q) = 2\tau_L^*(q) = e + 2sq - \frac{2\beta s^3 q}{(s + q)^3}$.
  Combining with \eqref{eq:H_app} yields:
  \begin{equation}
    e = D_2(q) := s(\alpha - c_p) - c_f + s\beta\bigl[\phi(q) + q\phi'(q)\bigr] - 5sq + \frac{2\beta s^3 q}{(s+q)^3}.
    \label{eq:D2_app}
  \end{equation}

  \smallskip
  \noindent\textbf{Scenario 4.}
  At symmetry, Corollary~\ref{cor:tax_rank} implies $T^{(4)}(q) = 2\tau_L^*(q) + \tau_G^*(q) = 2e + sq + \frac{\beta s q(3sq + q^2)}{(s + q)^3}$.
  Combining with \eqref{eq:H_app} yields:
  \begin{equation}
    e = D_4(q) := \frac{s}{2}\biggl[\alpha - c_p - \frac{c_f}{s} + \beta\bigl[\phi(q) + q\phi'(q)\bigr] - 4q - \frac{\beta q(3sq+q^2)}{(s+q)^3}\biggr].
    \label{eq:D4_app}
  \end{equation}
\end{lemma}

\begin{proof}
  \textbf{Scenario 2.}
  From $T^{(2)}(q) = e + 2sq - \frac{2\beta s^3 q}{(s + q)^3}$ and \eqref{eq:H_app},
  \begin{equation*}
    e + 2sq - \frac{2\beta s^3 q}{(s+q)^3} = s\bigl[\alpha - c_p + \beta\phi(q) + \beta q\phi'(q) - 3q\bigr] - c_f.
  \end{equation*}
  Rearranging yields \eqref{eq:D2_app}.

  \smallskip
  \noindent\textbf{Scenario 4.}
  From $T^{(4)}(q) = 2e + sq + \frac{\beta s q(3sq + q^2)}{(s + q)^3}$ and \eqref{eq:H_app},
  \begin{equation*}
    2e + sq + \frac{\beta sq(3sq+q^2)}{(s+q)^3} = s\bigl[\alpha - c_p + \beta\phi(q) + \beta q\phi'(q) - 3q\bigr] - c_f.
  \end{equation*}
  Divide by $2$ and rearrange to isolate $e$, obtaining \eqref{eq:D4_app}.
\end{proof}

\begin{lemma}[Welfare equivalence]\label{lem:deq_baseline_app}
  Let $q^{(1)*}$ denote the baseline equilibrium (Scenario~1) and define $\mathcal W_0(q) := \mathcal W(q; 0)$. For any $q \in (0, q^{(1)*})$, the unique damage level $e_{=}(q)$ such that $\mathcal W(q; e_{=}(q)) = \mathcal W(q^{(1)*}; e_{=}(q))$ is:
  \begin{equation}
    e_{=}(q) = \frac{s\big[\mathcal W_0(q^{(1)*}) - \mathcal W_0(q) \big]}{2 \big(q^{(1)*} - q\big)}.
    \label{eq:e_equalizer_app}
  \end{equation}
  Moreover, for any $e \geq 0$:
  \[
    \mathcal W(q; e) \gtrless \mathcal W(q^{(1)*}; e)
    \quad \Longleftrightarrow \quad
    e \gtrless e_{=}(q).
  \]
\end{lemma}
\begin{proof}
  From \eqref{eq:W_sym}, $\mathcal W(q; e) = \mathcal W_0(q) - \frac{2e}{s}q$. Therefore:
  \begin{equation*}
    \mathcal{W}(q;e) - \mathcal{W}(q^{(1)*};e) = \bigl[\mathcal{W}_0(q) - \mathcal{W}_0(q^{(1)*})\bigr] + \frac{2e}{s}\bigl(q^{(1)*} - q\bigr),
  \end{equation*}
  which is affine in $e$ with positive slope since $q^{(1)*} > q$. Solving the equality case yields \eqref{eq:e_equalizer_app}, and the inequality statement follows.
\end{proof}

\begin{lemma}[Welfare crossing]\label{lem:welfare_crossing}
  Let $\sigma \in \{2,4\}$. Define the welfare difference
  $\Delta_{\sigma}(e):=\mathcal W^{(\sigma)*}(e)-\mathcal W^{(1)*}(e)$.
  Under $\beta<s/2$:
  \begin{enumerate}
    \item[(i)] for $\Delta_{\sigma}(0)<0$ indicating zero environmental damage, regulation reduces welfare relative to the baseline,
    \item[(ii)] $\Delta_{\sigma}(e)$ is strictly increasing in $e$,
    \item[(iii)] there exists a unique $\bar e_{\sigma 1}>0$ such that $\Delta_{\sigma}(\bar e_{\sigma 1})=0$.
  \end{enumerate}
\end{lemma}
\begin{proof} From Appendix~\ref{app:rank_pos}, the symmetric Cournot condition is defined as $T=H(q)$ with $H$ strictly decreasing under Lemma~\ref{lem:cournot}'s stability condition. In Scenario~1, $T^{(1)}=0$ by definition, hence
  \begin{equation*}
    H\!\big(q^{(1)*}\big)=0.
  \end{equation*}
  In Scenario~2, under symmetry, the total charge is $T^{(2)}(q;e)=2\tau_L^*(q)$ with:
  \begin{equation*}
    2\tau_L^*(q) = e + 2sq - \frac{2\beta s^3 q}{(s+q)^3} \geq 2sq\Bigl(1 - \frac{\beta}{s}\Bigr) > 0, \quad \bigl(\text{since } (s+q)^3 \geq s^3 \text{ and } \beta < s/2\bigr).
  \end{equation*}
  so the Scenario~2 fixed point satisfies $T^{(2)*}(e)=H\!\big(q^{(2)*}(e)\big)>0$.
  In Scenario~4, Corollary \ref{cor:tax_rank} yields $T^{(4)}(q;e)=2\tau_L^*(q)+\tau_G^*(q)>0$ for all $q>0$, hence
  $T^{(4)*}(e)=H\!\big(q^{(4)*}(e)\big)>0$.
  Since $H$ is strictly decreasing and $H(q^{(1)*})=0$, it follows that:
  \[
    q^{(\sigma)*}(e)<q^{(1)*}\qquad\forall\,e\ge 0,\ \sigma \in \{2,4\}.
  \]

  \smallskip
  \noindent(i) Since $q^{(\sigma)*}(0)<q^{(1)*}$ and $\mathcal W_0(\cdot)$ is strictly concave with maximizer
  $q^{SO}(0)>q^{(1)*}$ (Lemma~\ref{lem:welfare_concavity} at $e=0$), then $\mathcal W_0$ is strictly increasing on
  $[0,q^{SO}(0)]$ hence:
  \[
    \mathcal W_0\!\big(q^{(\sigma)*}(0)\big)<\mathcal W_0\!\big(q^{(1)*}\big).
  \]
  Since $\mathcal W(q;0)=\mathcal W_0(q)$, this implies $\Delta_{\sigma}(0)<0$.

  \smallskip
  \noindent(ii) For each $e \ge 0$, define: \[g_e(q):=\mathcal W(q; e)-\mathcal W\!\big(q^{(1)*};e\big).\]
  Lemma~\ref{lem:deq_baseline_app} implies that for any fixed $q<q^{(1)*}$, the quantity $g_e(q)$ is strictly increasing in $e$. Since $q^{(\sigma)*}(e)<q^{(1)*}$ for all $e \ge 0$ and Lemma~\ref{lem:dq_maps_24_app} identifies $e=D_{\sigma}(q)$, where $D_{\sigma}$ is strictly decreasing in $q$ under $\beta<s/2$, then $q^{(\sigma)*}(e)=D_{\sigma}^{-1}(e)$ is strictly decreasing in $e$. Consequently, as $e$ increases, the evaluation point $q^{(\sigma)*}(e)$ shifts weakly to the left, while for every admissible $q<q^{(1)*}$ the value $g_e(q)$ shifts strictly upward. Therefore:
  \[
    \Delta_{\sigma}(e)=g_e\!\big(q^{(\sigma)*}(e)\big)
  \]
  is strictly increasing in $e$.

  \smallskip
  \noindent (iii) From (i), $\Delta_{\sigma}(0)<0$. From (ii), $\Delta_{\sigma}$ is strictly increasing and continuous in $e$.
  Moreover, $\mathcal W(q; e)=\mathcal W_0(q)-\frac{2e}{s}q$ under symmetry hence:
  \begin{equation*}
    \Delta_{\sigma}(e) = \bigl[\mathcal{W}_0(q^{(\sigma)*}(e)) - \mathcal{W}_0(q^{(1)*})\bigr] + \frac{2e}{s}\bigl[q^{(1)*} - q^{(\sigma)*}(e)\bigr].
  \end{equation*}
  $q^{(1)*}-q^{(\sigma)*}(e)\ge q^{(1)*}-q^{(\sigma)*}(0)>0$, hence the linear term in $e$ implies
  $\lim_{e\to\infty}\Delta_{\sigma}(e)=+\infty$. The intermediate value theorem ensures a unique (by strict monotonicity) $\bar e_{\sigma 1}>0$ such that $\Delta_{\sigma}(\bar e_{\sigma 1})=0$.

\end{proof}

\begin{proposition}[Crossing condition]\label{prop:crossing_system_app}
  Let $\sigma \in \{2,4\}$. Any welfare crossing $e = \bar e_{\sigma 1}$ between Scenario~$\sigma$ and Scenario~1 occurs at some $q_{\sigma 1} \in (0, q^{(1)*})$ satisfying the system
  \begin{equation}
    \bar e_{\sigma 1} = e_{=}(q_{\sigma 1}) = D_{\sigma}(q_{\sigma 1}),
    \label{eq:crossing_system_app}
  \end{equation}
  where $e_{=}(q)$ is given by \eqref{eq:e_equalizer_app} and $D_{\sigma}(q)$ by \eqref{eq:D2_app} or \eqref{eq:D4_app}.
\end{proposition}

\begin{proof}
  At the crossing $e = \bar e_{\sigma 1}$, let $q_{\sigma 1} := q^{(\sigma)*}(\bar e_{\sigma 1})$ denote the equilibrium quantity in Scenario~$\sigma$. The welfare equality $\mathcal W(q_{\sigma 1}; \bar e_{\sigma 1}) = \mathcal W(q^{(1)*}; \bar e_{\sigma 1})$ is equivalent to $\bar e_{\sigma 1} = e_{=}(q_{\sigma 1})$ by Lemma~\ref{lem:deq_baseline_app}. The equilibrium condition requires $\bar e_{\sigma 1} = D_{\sigma}(q_{\sigma 1})$ by Lemma~\ref{lem:dq_maps_24_app}. Combining yields \eqref{eq:crossing_system_app}.
\end{proof}

\subsection{Proof of Lemma~\ref{lem:e_threshold_neg} and Proposition~\ref{prop:subsidy_regime}}
\label{app:subsidy_regime}

\begin{proof}[Proof of Lemma~\ref{lem:e_threshold_neg}.]

  Let $q^{SO}(e)$ denote the unique maximum of $\mathcal W(\cdot; e)$ (Lemma~\ref{lem:welfare_concavity}). It satisfies:
  \[
    p(q) = \alpha + \beta \phi(q) - 2q = c_p + \frac{c_f + e}{s}.
  \]
  Since $p'(q) = \beta \phi'(q) - 2 < 0$ under $\beta < s/2$ and Lemma~\ref{lem:phi_app}, the implicit-function theorem implies:
  \[
    \frac{d q^{SO}(e)}{d e} = \frac{(1/s)}{p'(q^{SO}(e))} < 0,
  \]
  so $q^{SO}(e)$ is strictly decreasing in $e$.

  Define $\bar e_G := s q^{(1)*}[1 - \beta \phi'(q^{(1)*})]$. At the baseline equilibrium, the markup identity \eqref{eq:markup_app} with $T = 0$ gives:
  \[
    p^{(1)*} - \Big(c_p + \frac{c_f}{s}\Big) = q^{(1)*} \big[1 - \beta \phi'(q^{(1)*}) \big] = \frac{\bar e_G}{s}.
  \]
  Therefore:
  \[
    c_p + \frac{c_f + \bar e_G}{s} = c_p + \frac{c_f}{s} + \frac{\bar e_G}{s} = p^{(1)*}.
  \]
  However the first-best condition is $p \big(q^{SO}(\bar e_G) \big) = c_p + \frac{c_f + \bar e_G}{s} = p^{(1)*}$. Since $p(\cdot)$ is strictly decreasing, this implies $q^{SO}(\bar e_G) = q^{(1)*}$.

  Scenario~3 implements $q^{SO}(e)$ (Proposition~\ref{prop:scenario3_SO}), so $q^{(3)*}(e) = q^{SO}(e)$ hence $q^{(3)*}(e) \gtrless q^{(1)*}$ iff $e \lessgtr \bar e_G$.

  Finally, at any symmetric output $q$, the global best-response charge is $\tau_G^*(q) = e - sq[1 - \beta \phi'(q)]$. Evaluating at the Scenario~3 fixed point $q^{(3)*}(e)$ yields
  \[
    \tau_G^{(3)*}(e) = e - sq^{(3)*}(e) \Big[1 - \beta \phi'(q^{(3)*}(e))\Big].
  \]
  At $e = \bar e_G$, $q^{(3)*}(\bar e_G) = q^{(1)*}$ so $\tau_G^{(3)*}(\bar e_G) = 0$. Define $h(q) := sq[1 - \beta\phi'(q)]$. Since $\phi'' < 0$ (Lemma~\ref{lem:phi_app}) and $1 - \beta\phi'(q) > 0$ under $\beta < s/2$, we have $h'(q) > 0$, so $h$ is strictly increasing. For $e < \bar e_G$, $q^{(3)*}(e) > q^{(1)*}$, hence $\tau_G^{(3)*}(e) = e - h(q^{(3)*}(e)) < \bar e_G - h(q^{(1)*}) = 0.$ Similarly $e > \bar e_G$ implies $\tau_G^{(3)*}(e) > 0$.
\end{proof}

\begin{proof}[Proof of Proposition~\ref{prop:subsidy_regime}.]

  Under $e < \bar e_G$, Lemma~\ref{lem:e_threshold_neg} implies $e < \bar e_G \Rightarrow \tau_G^{(3)*}(e) < 0$, hence $T^{(3)*}(e) < 0$. Scenario~1 has $T^{(1)} = 0$. In Scenario~2, $T^{(2)*}(e) = 2 \tau_L^*(q^{(2)*})$ and $\tau_L^*(q) > 0$ for $e \geq 0$ and $q > 0$ from \eqref{eq:tauL} (at $e = 0$, $\tau_L^*(q) = sq[1 - \beta s^2/(s+q)^3] > 0$ since $\beta < s/2$), hence $T^{(2)*}(e) > 0$. In Scenario~4, $T^{(4)*}(e) > 0$ for all $q > 0$ by Corollary~\ref{cor:tax_rank}.

  Since symmetric Cournot quantity decreases in the scalar instrument (Lemma~\ref{lem:cournot}), we obtain:
  \[
    q^{(2)*}(e) < q^{(1)*} < q^{(3)*}(e),
    \qquad
    q^{(4)*}(e) < q^{(1)*} < q^{(3)*}(e),
  \]
  and the reverse ordering for prices.

  For welfare, Scenario~3 implements the first-best and strictly dominates all other regimes. Since $q^{(1)*}$ lies to the left of the maximum $q^{(3)*}(e)$ when $e < \bar e_G$ (Lemma~\ref{lem:e_threshold_neg}), strict concavity implies $\mathcal W(\cdot; e)$ is increasing on $[0, q^{(3)*}(e)]$. Therefore $q^{(2)*}(e), q^{(4)*}(e) < q^{(1)*}$ implies:
  \begin{equation*}
    \mathcal{W}^{(1)*}(e) = \mathcal{W}(q^{(1)*}; e)
    > \mathcal{W}(q^{(2)*}(e); e)
    = \mathcal{W}^{(2)*}(e),
    \qquad
    \mathcal{W}^{(1)*}(e) > \mathcal{W}^{(4)*}(e).
  \end{equation*}
  which yields \eqref{eq:W_rank_neg}.

  Finally, because $\mathcal W(\cdot; e)$ is increasing on the relevant interval, $\mathcal W^{(4)*}(e) \gtrless \mathcal W^{(2)*}(e)$ iff $q^{(4)*}(e) \gtrless q^{(2)*}(e)$. Moreover, using $T = H(q)$ and $T^{(4)}(q) = T^{(2)}(q) + \tau_G^*(q)$ obtains:
  \begin{equation*}
    q^{(4)*}(e) \gtrless q^{(2)*}(e)
    \iff T^{(4)}(q^{(2)*}(e)) \lessgtr T^{(2)}(q^{(2)*}(e))
    \iff \tau_G^*(q^{(2)*}(e)) \lessgtr 0,
  \end{equation*}
\end{proof}

\section{Algorithm}
\label{app:algorithm}

This appendix presents the pseudocode of the algorithm used to solve the two-stage game and identify the corresponding subgame-perfect Nash equilibria.

\begin{algorithm}[htbp]
  \caption{Find the Sub-game Perfect Nash Equilibria (SPNE) of $\Gamma$}
  \label{alg:algorithm}
  \resizebox{1\textwidth}{!}{%
    \begin{minipage}{1\textwidth}
      \footnotesize
      \begin{algorithmic}[1]
        \setlength{\itemsep}{-1pt}
        \setlength{\parsep}{0pt}
        \setlength{\topsep}{0pt}
        \setlength{\partopsep}{0pt}

        \State \textbf{Start}
        \State Initialize regulators' and airlines' decision variables and networks.
        \State Set grid radius $\rho_1>0$ and convergence thresholds $threshold_1, threshold_2$.
        \State Define a set of starting points $\mathcal{S}$ and a set of airline update orders $\mathcal{O}$.
        \State $solution_1^0 \leftarrow [0,\dots,0]$ of length $|\mathcal{R}|$.
        \State $solution_1 \leftarrow$ initial first-stage decisions, one entry per regulator.

        \While{$\lVert solution_1 - solution_1^0 \rVert > threshold_1$}
        \State $solution_1^0 \leftarrow solution_1$.

        \For{each regulator $r \in \mathcal{R}$}
        \State Create a point grid $\mathcal{G}_r$ around $solution_1[r]$ with radius $\rho_1$.

        \For{each point $g \in \mathcal{G}_r$}
        \State Fix regulator $r$ at $g$ and all other regulators at their current values in $solution_1$.
        \State Reset airlines' decision variables and networks for this point $g$.
        \State Initialize an empty set $\mathcal{E}(g)$ to store second-stage outcomes.

        \For{each $(s,o) \in \mathcal{S}\times\mathcal{O}$}
        \State Initialize airlines' decision variables using starting point $s$.
        \State Set the airline update order according to $o$.
        \State $solution_2^0 \leftarrow [0,\dots,0]$ of length $|\mathcal{A}|$.
        \State $solution_2 \leftarrow$ empty vector of length $|\mathcal{A}|$.

        \For{each airline $a \in \mathcal{A}$ following order $o$}
        \State Solve mathematical program using \textit{IPOPT} and append to $solution_2$.
        \EndFor

        \While{$\lVert solution_2 - solution_2^0 \rVert > threshold_2$}
        \State $solution_2^0 \leftarrow solution_2$.

        \For{each airline $a \in \mathcal{A}$ following order $o$}
        \State Solve mathematical program using \textit{IPOPT} and update entry in $solution_2$.
        \EndFor
        \EndWhile

        \State Store $solution_2$ in $\mathcal{E}(g)$.
        \EndFor

        \State Evaluate welfare $W_r(g)$ under outcomes in $\mathcal{E}(g)$.
        \EndFor

        \State Select $g^\star \in \mathcal{G}_r$ that maximizes $W_r(g)$.
        \State Update $solution_1[r] \leftarrow g^\star$.
        \EndFor

        \State Shrink the grid radius $\rho_1 \leftarrow \alpha \rho_1$ with $0<\alpha<1$.
        \EndWhile

        \State \Return first-stage decisions $solution_1$ and corresponding second-stage outcomes.
      \end{algorithmic}
    \end{minipage}%
  }
\end{algorithm}
\FloatBarrier

\bibliography{references}

\end{document}

%% file: Notation.tex
\begin{tabular}{rl}
    \hline
    \multicolumn{2}{l}{\textbf{Sets and Indices}}                                                                           \\
    $\mathcal{A}$          & Set of airlines; indexed by $a$                                                                \\
    $\mathcal{H}$          & Set of the type of flight (i.e. short and long haul flight); indexed by $h$                    \\
    $\mathcal{K}$          & Set of all legs in the network according to the flight type $h$; indexed by $k^{h}$            \\
    $\mathcal{N}$          & Set of airports nodes; indexed by $i$, $j$                                                     \\
    $\mathcal{R}$          & Set of regulators; indexed by $r$                                                              \\
    $\mathcal{T}$          & Set of passenger types; indexed by $t$                                                         \\
    $\mathcal{V}$          & Set of aircraft versions; indexed by $v$                                                       \\
    \multicolumn{2}{l}{\textbf{Parameters}}                                                                                 \\
    $\beta_{\cdot t}$      & parameter on the corresponding characteristic in the utility function for passenger type $t$   \\
    $c_{ka}$               & Cost for airline $a$ to serve leg $k$                                                          \\
    $d_{ijt}$              & Demand between nodes $i$ and $j$ for passenger type $t$                                        \\
    $\delta_{ija}$         & 1 if the connection between $i$ and $j$ operated by $a$ is direct, 0 otherwise                 \\
    $e$                    & Social cost of carbon emissions                                                                \\
    $\epsilon_{ijta}$      & Random component of utility between nodes $i$ and $j$ for passenger type $t$ and airline $a$   \\
    $\varepsilon_{kv}$     & Tons of $CO_2$ burnt per km on leg $k$ per aircraft version $v$                                \\
    $\bar{f}_h$            & Average utilization of aircraft in the time period by type of flight $h$                       \\
    $g_{k}$                & Great circle distance of leg $k$                                                               \\
    $s_{kv}$               & Seats available on flights served by $a$ on leg $k$ per aircraft version $v$                   \\
    $\phi_{hv}$            & Fuel burnt per kilometer by aircraft version $v$ type $h$                                      \\
    $\bar{p}$              & Fuel price in dollars per ton of kerosene                                                      \\
    $o_{hv}$               & Ownership cost for aircraft type $h$ version $v$                                               \\
    $co_2$                 & Conversion factor between fuel consumption and $co_2$ produced                                 \\
    $\chi$                 & Share of operating costs without fuel and ownership costs                                      \\
    $\lambda_{hv}$         & Salvage value of an aircraft of type $h$ version $v$ at the end of the time period             \\
    $\ell$                 & Interest rate                                                                                  \\
    $n$                    & Time periods                                                                                   \\
    $\Psi_{hv}$            & Initial purchase price of an aircraft of type $h$ version $v$                                  \\
    $\rho_t$               & Heterogeneity between nests in the NL model for passenger type $t$                             \\
    \multicolumn{2}{l}{\textbf{Decision variables}}                                                                         \\
    $\theta_{r}$           & Carbon tax imposed by regulator $r$                                                            \\
    $f_{ka}$               & Service frequency on leg $k$ for airline $a$                                                   \\
    $p_{ijta}$             & Fare set for itinerary from $i$ to $j$ and passenger type $t$  for airline $a$                 \\
    $\Tilde{\alpha}_{hva}$ & Number of aircraft owned of type $h$ version $v$    for airline $a$                            \\
    \multicolumn{2}{l}{\textbf{Auxiliary variables}}                                                                        \\
    $m_{ijta}$             & Market share of airline $a$ for itinerary from $i$ to $j$ and passenger type $t$               \\
    $V_{ijta}$             & Systematic component of utility between nodes $i$ to $j$ per passenger type $t$ of airline $a$ \\
    $z_{ija}$              & Minimum frequency over an indirect itinerary from $i$ to $j$ for airline $a$                   \\
    \hline
\end{tabular}

%% file: aircraft_costs.tex
{
    \begin{threeparttable}
        \setlength{\tabcolsep}{6pt}
        \renewcommand{\arraystretch}{1.15}
        \footnotesize
        \begin{tabular}{
            ll
            S[table-format=3.0]
            S[table-format=3.2]
            S[table-format=1.3]
            S[table-format=3.0]
            S[table-format=3.0]
        }
            \toprule
            \multicolumn{2}{c}{Aircraft type} &
            \multicolumn{1}{c}{Purchase} &
            \multicolumn{1}{c}{Salvage} &
            \multicolumn{1}{c}{Fuel} &
            \multicolumn{1}{c}{Seats} &
            \multicolumn{1}{c}{Seats} \\
            \multicolumn{2}{c}{} &
            \multicolumn{1}{c}{(\$M)} &
            \multicolumn{1}{c}{(\$M)} &
            \multicolumn{1}{c}{(ton/km)} &
            \multicolumn{1}{c}{Legacy} &
            \multicolumn{1}{c}{LCC} \\
            \midrule
            \multirow{2}{*}{Short-haul} & A320ceo & 100 & 33.33  & 0.004 & 151 & 187 \\
                       & A320neo & 110 & 36.30  & 0.003 & 151 & 187 \\
            \addlinespace
            \multirow{2}{*}{Long-haul}  & B777    & 345 & 113.85 & 0.008 & 400 & \multicolumn{1}{c}{--} \\
                       & A350    & 370 & 122.10 & 0.007 & 400 & \multicolumn{1}{c}{--} \\
            \bottomrule
        \end{tabular}
    \end{threeparttable}
    }

%% file: real_world.tex
\begin{threeparttable}
\setlength{\tabcolsep}{6pt}
\small
\begin{tabular}{lrrrr}
\toprule
 & European legacies & European LCCs & American legacies & American LCCs \\
\midrule
\addlinespace[2pt]
\multicolumn{5}{l}{\emph{Passengers (millions)}}\\
Passengers, short-haul economy  & 135.56 (126.98) & 282.57 (427.86) & 478.76 (326.61) & 212.21 (165.31) \\
Passengers, short-haul business & 11.44 (15.44) & \textemdash & 12.58 (29.74) & \textemdash \\
Passengers, long-haul economy   & 28.60 (40.61) & \textemdash & 65.78 (57.20) & \textemdash \\
Passengers, long-haul business  & 2.29 (9.15) & \textemdash & 5.72 (5.15) & \textemdash \\
\addlinespace[4pt]
\multicolumn{5}{l}{\emph{Unit costs and revenues (\euro{} cents)}}\\
Average CASK                    & 5 (7)       & 3 (4)       & 5 (7)       & 3 (5)       \\
Average RASK                    & 14 (8)      & 8 (5)       & 13 (8)      & 7 (6)       \\
\addlinespace[4pt]
\multicolumn{5}{l}{\emph{Average fares (\euro{})}}\\
Short-haul economy              & 182 (152)   & 144 (82)    & 237 (204)   & 218 (174)   \\
Short-haul business             & 862 (544)   & \textemdash & 3{,}001 (814) & \textemdash \\
Long-haul economy               & 491 (258)   & \textemdash & 579 (228)   & \textemdash \\
Long-haul business              & 8{,}393 (3{,}500) & \textemdash & 8{,}933 (3{,}331) & \textemdash \\
\addlinespace[4pt]
\multicolumn{5}{l}{\emph{Frequencies (thousands of flights)}}\\
Total frequency, short-haul     & 1{,}345.34 (1{,}594.74) & 2{,}500.78 (3{,}081.94) & 4{,}138.42 (5{,}044.47) & 1{,}431.14 (2{,}218.22) \\
Total frequency, long-haul      & 105.25 (187.04)   & \textemdash    & 232.23 (331.19)   & \textemdash \\
\bottomrule
\end{tabular}
\begin{tablenotes}[flushleft]\footnotesize
\item \emph{Notes:} Values in brackets correspond to real-world data. Discrepancies arise from the omission of connecting passengers when using a subset of the complete network and from the challenge of reproducing, within a single model and representative network, the differing cost structures of regular and low-cost carriers.
\end{tablenotes}
\end{threeparttable}

%% file: baseline_tab.tex
\robustify\bfseries
\sisetup{
  detect-weight=true,
  detect-family=true,
  retain-explicit-plus
}
\begin{threeparttable}
\setlength{\tabcolsep}{5pt}
\small
\begin{tabular}{
  l
  l
  S[table-format=+3.0]
  S[table-format=+3.2]
  S[table-format=+3.2]
  S[table-format=+3.2]
  S[table-format=+3.2]
  S[table-format=+3.2]
}
\toprule
 & & {\textbf{Charges}} & \multicolumn{5}{c}{\textbf{Welfare outcomes (\,\euro~B)}} \\
\cmidrule(lr){3-3}\cmidrule(lr){4-8}
\textbf{Scenario} & \textbf{Region} & {(\euro/ton)} & {Welfare} & {Emissions} & {Producers} & {Consumers} & {Government} \\
\midrule

\multirow{2}{*}{Baseline}
  & Europe   & 0  & 155.27 & -46.33 & 56.63  & 144.72 & 0.00 \\
  & America  & 0  & 300.87 & -97.81 & 140.71 & 257.97 & 0.00 \\
\cmidrule(l){2-8}
  & \textbf{Total ($\mathcal{W}$)} &  & \bfseries 456.14 & \bfseries -144.14 & \bfseries 197.34 & \bfseries 402.69 & \bfseries 0.00 \\
\bottomrule
\end{tabular}
\end{threeparttable}

%% file: Airlines_results.tex
\begin{threeparttable}
  \setlength{\tabcolsep}{6pt}
  \small
  \begin{tabular}{
      ll
      S[table-format=3.0]  
      S[table-format=2.0]  
      S[table-format=3.0]  
      S[table-format=4.0]  
      S[table-format=4.0]  
      S[table-format=4.0]  
      S[table-format=4.0]  
      S[table-format=4.0]  
    }
    \toprule
             &                & \multicolumn{2}{c}{Passengers} & \multicolumn{2}{c}{Prices}  & \multicolumn{2}{c}{Frequency}          & \multicolumn{2}{c}{Fleet}                                      \\
             &                & \multicolumn{2}{c}{(million)}  & \multicolumn{2}{c}{(\euro)} & \multicolumn{2}{c}{(thousand flights)} & \multicolumn{2}{c}{(aircraft)}                                 \\
    \cmidrule(lr){3-4}\cmidrule(lr){5-6}\cmidrule(lr){7-8}\cmidrule(lr){9-10}
    Scenario & Region         & {Economy}                      & {Business}                  & {Economy}                              & {Business}                     & {New} & {Old} & {New} & {Old} \\
    \midrule

    \multirow{3}{*}{Baseline}
             & America        & 691                            & 13                          & 226                                    & 2,986                          & 0     & 5570  & 0     & 7651  \\
             & Europe         & 418                            & 11                          & 156                                    & 867                            & 0     & 3846  & 0     & 2603  \\
             & Trans-Atlantic & 94                             & 7                           & 554                                    & 8,765                          & 0     & 338   & 0     & 930   \\
    \addlinespace

    \multirow{3}{*}{Local regulators}
             & America        & 665                            & 13                          & 233                                    & 2,990                          & 1742  & 3605  & 2388  & 4948  \\
             & Europe         & 395                            & 11                          & 160                                    & 871                            & 2738  & 665   & 1730  & 615   \\
             & Trans-Atlantic & 87                             & 7                           & 584                                    & 8,793                          & 76    & 221   & 215   & 601   \\
    \addlinespace

    \multirow{3}{*}{\shortstack[l]{Uniform \\ global regulation}}
             & America        & 637                            & 12                          & 241                                    & 2,996                          & 4048  & 1073  & 5563  & 1473  \\
             & Europe         & 391                            & 11                          & 160                                    & 872                            & 2702  & 631   & 1730  & 572   \\
             & Trans-Atlantic & 83                             & 7                           & 603                                    & 8,809                          & 244   & 37    & 672   & 100   \\
    \addlinespace

    \multirow{3}{*}{\shortstack[l]{Differentiated\\ global regulation}}
             & America        & 637                            & 12                          & 241                                    & 2,996                          & 4048  & 1073  & 5563  & 1473  \\
             & Europe         & 395                            & 11                          & 160                                    & 871                            & 2737  & 665   & 1730  & 615   \\
             & Trans-Atlantic & 84                             & 7                           & 599                                    & 8,806                          & 199   & 86    & 544   & 229   \\
    \addlinespace

    \multirow{3}{*}{Overlapping}
             & America        & 652                            & 12                          & 237                                    & 2,993                          & 2925  & 2318  & 4018  & 3189  \\
             & Europe         & 394                            & 11                          & 160                                    & 871                            & 2723  & 663   & 1730  & 601   \\
             & Trans-Atlantic & 85                             & 7                           & 591                                    & 8,799                          & 117   & 173   & 329   & 472   \\
    \bottomrule
  \end{tabular}
\end{threeparttable}

%% file: REGs_results.tex
\robustify\bfseries
\sisetup{
  detect-weight=true,
  detect-family=true,
  retain-explicit-plus
}
\begin{threeparttable}
\setlength{\tabcolsep}{5pt}
\small
\begin{tabular}{
  l
  l
  S[table-format=+3.0] 
  S[table-format=+2.2] 
  S[table-format=+2.2] 
  S[table-format=+2.2] 
  S[table-format=+2.2] 
  S[table-format=+2.2] 
}
\toprule
 & & {\textbf{Charges}} & \multicolumn{5}{c}{\textbf{Welfare variations (\,\euro~B)}} \\
\cmidrule(lr){3-3}\cmidrule(lr){4-8}
\textbf{Scenario} & \textbf{Region} & {(\euro/ton)} & {Welfare} & {Emissions} & {Producers} & {Consumers} & {Government} \\
\midrule

\multirow{2}{*}{Baseline}
  & Europe   & 0  & 0 & 0 & 0 & 0 & 0 \\
  & America  & 0  & 0 & 0 & 0 & 0 & 0 \\
  \cmidrule(l){2-8}
  & \textbf{Total ($\mathcal{W}$)} &  & \bfseries 0 & \bfseries 0 & \bfseries 0$^{*}$ & \bfseries 0$^{*}$ & \bfseries 0 \\
\addlinespace[1.5ex]

\multirow{2}{*}{\shortstack[l]{Local\\regulators}}
  & Europe   & 51 & +8.13 & +13.61 & -6.62 & -7.45 & +8.59 \\
  & America  & 29 & +3.08 & +13.61 & -12.85 & -9.69 & +12.02 \\
  \cmidrule(l){2-8}
  & \textbf{Total ($\mathcal{W}$)} &  & \bfseries +11.21 & \bfseries +27.22 & \bfseries -19.47 & \bfseries -17.14 & \bfseries +20.62 \\
\addlinespace[1.5ex]

\multirow{3}{*}{\shortstack[l]{Uniform \\ global regulation}}
  & Europe   &  & +18.01 & +20.34 & -8.03 & -9.78 & +15.48 \\
  & America  &  & -4.27 & +20.34 & -22.06 & -18.02 & +15.48 \\
  & Global   & 60   &   &   &   &   &  \\
  \cmidrule(l){2-8}
  & \textbf{Total ($\mathcal{W}$)} &  & \bfseries +13.74 & \bfseries +40.68$^{*}$ & \bfseries -30.09 & \bfseries -27.08 & \bfseries +30.96$^{*}$ \\
\addlinespace[1.5ex]

\multirow{3}{*}{\shortstack[l]{Differentiated\\ global regulation}}
  & Europe   &  & +18.76 & +19.71 & -7.18 & -8.68 & +14.91 \\
  & America  &  & -4.76 & +19.71 & -21.69 & -17.69 & +14.91 \\
  & Global   & \text{EU: 51 \& NA: 60} &  &  &  &   &  \\
  \cmidrule(l){2-8}
  & \textbf{Total ($\mathcal{W}$)} &  & \bfseries +13.98$^{*}$ & \bfseries +39.42 & \bfseries -28.87 & \bfseries -26.37 & \bfseries +29.81 \\
\addlinespace[1.5ex]

\multirow{3}{*}{Overlapping}
  & Europe   & 0   & +16.13 & +16.92 & -7.08 & -8.28 & +14.57 \\
  & America  & -10 & -2.93 & +16.92 & -17.13 & -13.46 & +10.75 \\
  & Global\tnote{$\dagger$} & 53  &  &  &  &  & 0 \\
  \cmidrule(l){2-8}
  & \textbf{Total ($\mathcal{W}$)} &  & \bfseries +13.20 & \bfseries +33.84 & \bfseries -24.22 & \bfseries -21.74 & \bfseries +25.32 \\

\bottomrule
\end{tabular}
\begin{tablenotes}[flushleft]\footnotesize
\item[$^{*}$] Associated with the optimal scenario for each welfare component.
\item[$\dagger$] The "Global" row represents a distinct global regulator redistributing half of the tax revenues to each of the two regions.
\end{tablenotes}
\end{threeparttable}

%% file: airline_response_summary.tex
\begin{tabular}{
  l
  S[table-format=-2.2]
  S[table-format=-1.2]
  S[table-format=+2.2]
  S[table-format=-1.2]
  S[table-format=+2.2]
}
\toprule
Scenario
& \multicolumn{1}{c}{Frequency $\Delta$}
& \multicolumn{1}{c}{Fleet $\Delta$}
& \multicolumn{1}{c}{Fare $\Delta$}
& \multicolumn{1}{c}{Passengers $\Delta$}
& \multicolumn{1}{c}{Routes affected} \\
&
\multicolumn{1}{c}{(\%)}
& \multicolumn{1}{c}{(\%)}
& \multicolumn{1}{c}{(\euro)}
& \multicolumn{1}{c}{(\%)}
& \multicolumn{1}{c}{(\%)} \\
\midrule
Local regulators                  &  -7.26 & -6.08 &  5.30 & -4.65 & 23.75 \\
Uniform global regulation         & -10.43 & -9.59 & 10.15 & -7.61 & 25.49 \\
Differentiated global regulation  &  -9.70 & -9.16 &  9.70 & -7.19 & 25.04 \\
Overlapping                       &  -8.56 & -7.62 &  7.36 & -5.91 & 24.58 \\
\bottomrule
\end{tabular}

%% file: geographic_incidence.tex
\begin{tabular}{
  l
  l
  S[table-format=-2.2] 
  S[table-format=+2.2] 
  S[table-format=-2.2] 
  S[table-format=+2.2] 
}
\toprule
Scenario & Region
& \multicolumn{1}{c}{Frequency $\Delta$}
& \multicolumn{1}{c}{Fare $\Delta$}
& \multicolumn{1}{c}{Passengers $\Delta$}
& \multicolumn{1}{c}{Routes affected} \\
&
& \multicolumn{1}{c}{(\%)}
& \multicolumn{1}{c}{(\euro)}
& \multicolumn{1}{c}{(\%)}
& \multicolumn{1}{c}{(\%)} \\
\midrule
\multirow{3}{*}{\shortstack[l]{Local\\ regulators}} 
& Europe         & -11.54 &  1.33 & -5.43 & 19.48 \\
& North America  &  -4.00 &  5.05 & -3.71 & 15.13 \\
& Transatlantic  & -12.15 & 23.71 & -7.75 & 52.38 \\

\addlinespace[1.5ex]

\multirow{3}{*}{\shortstack[l]{Uniform\\ global regulation}}
& Europe         & -13.33 &  1.60 &  -6.40 & 20.56 \\
& North America  &  -8.05 & 10.96 &  -7.74 & 16.78 \\
& Transatlantic  & -16.61 & 40.47 & -11.72 & 55.56 \\

\addlinespace[1.5ex]

\multirow{3}{*}{\shortstack[l]{Differentiated\\ global regulation}} 
& Europe         & -11.56 &  1.25 &  -5.40 & 19.48 \\
& North America  &  -8.05 & 10.97 &  -7.75 & 16.78 \\
& Transatlantic  & -15.74 & 36.55 & -10.86 & 55.16 \\

\addlinespace[1.5ex]

\multirow{3}{*}{\shortstack[l]{Overlapping}} 
& Europe         & -11.95 &  1.37 & -5.64 & 19.91 \\
& North America  &  -5.87 &  7.72 & -5.57 & 15.95 \\
& Transatlantic  & -14.16 & 30.07 & -9.39 & 53.97 \\
\bottomrule
\end{tabular}

%% file: airlines_business_models.tex
\begin{tabular}{
  l
  l
  l
  S[table-format=-2.2]
  S[table-format=-2.2]
  S[table-format=+2.2]
  S[table-format=-1.2]
}
\toprule
Scenario & Region & Type
& \multicolumn{1}{c}{Frequency $\Delta$}
& \multicolumn{1}{c}{Fleet $\Delta$}
& \multicolumn{1}{c}{Fare $\Delta$}
& \multicolumn{1}{c}{Passengers $\Delta$} \\
&
&
& \multicolumn{1}{c}{(\%)}
& \multicolumn{1}{c}{(\%)}
& \multicolumn{1}{c}{(\euro)}
& \multicolumn{1}{c}{(\%)} \\
\midrule
\multirow{4}{*}{\shortstack[l]{Local\\ regulators}}
& \multirow{2}{*}{American}
  & LCC    &  -3.64 &  -3.64 &  3.76 & -3.11 \\
& & Legacy &  -4.47 &  -4.79 &  7.96 & -4.39 \\
\addlinespace
& \multirow{2}{*}{European}
  & LCC    & -15.80 & -15.80 &  0.11 & -7.14 \\
& & Legacy &  -4.48 &  -5.85 &  6.99 & -3.32 \\
\cmidrule(l){2-7}

\multirow{4}{*}{\shortstack[l]{Uniform\\ global regulation}}
& \multirow{2}{*}{American}
  & LCC    &  -8.04 &  -8.04 &  9.17 & -7.43 \\
& & Legacy &  -8.41 &  -8.73 & 15.39 & -8.28 \\
\addlinespace
& \multirow{2}{*}{European}
  & LCC    & -18.18 & -18.18 &  0.23 & -8.43 \\
& & Legacy &  -5.51 &  -7.42 & 10.48 & -4.40 \\
\cmidrule(l){2-7}

\multirow{4}{*}{\shortstack[l]{Discriminating\\ global regulation}}
& \multirow{2}{*}{American}
  & LCC    &  -8.04 &  -8.04 &  9.17 & -7.43 \\
& & Legacy &  -8.37 &  -8.65 & 14.92 & -8.18 \\
\addlinespace
& \multirow{2}{*}{European}
  & LCC    & -15.81 & -15.81 &  0.10 & -7.16 \\
& & Legacy &  -4.82 &  -6.70 &  9.13 & -3.83 \\
\cmidrule(l){2-7}

\multirow{4}{*}{Overlapping}
& \multirow{2}{*}{American}
  & LCC    &  -5.52 &  -5.52 &  6.11 & -4.94 \\
& & Legacy &  -6.34 &  -6.66 & 11.21 & -6.22 \\
\addlinespace
& \multirow{2}{*}{European}
  & LCC    & -16.34 & -16.34 &  0.13 & -7.44 \\
& & Legacy &  -4.81 &  -6.45 &  8.15 & -3.68 \\

\bottomrule
\end{tabular}

%% file: logit.tex
\begin{threeparttable}
\small
\setlength{\tabcolsep}{5pt}
\begin{tabular}{lcccccc}
\toprule
 & \multicolumn{2}{c}{Europe} & \multicolumn{2}{c}{North America} & \multicolumn{2}{c}{Transatlantic} \\
\cmidrule(lr){2-3}\cmidrule(lr){4-5}\cmidrule(lr){6-7}
 & Economy & Business & Economy & Business & Economy & Business \\
\midrule
Constant                    & $-9.80^{**}$  & $-14.18^{**}$ & $-9.38^{**}$  & $-13.38^{**}$ & $-10.21^{**}$ & $-12.31^{**}$ \\
                            & (0.06)        & (0.11)        & (0.06)        & (0.05)        & (0.13)        & (0.15)        \\
Market distance (1{,}000 km) & $0.54^{**}$   & $0.76^{**}$   & $0.65^{**}$   & $0.38^{**}$   & $-0.13^{**}$  & $-0.23^{**}$  \\
                            & (0.02)        & (0.08)        & (0.01)        & (0.02)        & (0.03)        & (0.04)        \\
Market distance$^{2}$ (1{,}000 km)$^{2}$ 
                            & $-0.15^{**}$  & $-0.10^{**}$  & $-0.10^{**}$  & $-0.04^{**}$  & $0.02^{**}$   & $0.02^{**}$   \\
                            & (0.01)        & (0.02)        & (0.002)       & (0.004)       & (0.002)       & (0.002)       \\
GDP per capita              & $0.40^{**}$   & $1.91^{**}$   & $0.61^{**}$   & $0.58^{**}$   & $0.68^{**}$   & $1.34^{**}$   \\
                            & (0.04)        & (0.11)        & (0.03)        & (0.05)        & (0.04)        & (0.05)        \\
Directness (0/1)            & $3.11^{**}$   & $3.10^{**}$   & $2.24^{**}$   & $2.55^{**}$   & $2.84^{**}$   & $2.65^{**}$   \\
                            & (0.03)        & (0.03)        & (0.02)        & (0.02)        & (0.03)        & (0.02)        \\
Detour distance (1{,}000 km) & $-0.29^{**}$  & $-0.13^{**}$  & $-0.28^{**}$  & $-0.21^{**}$  & $-0.18^{**}$  & $-0.13^{**}$  \\
                            & (0.01)        & (0.02)        & (0.01)        & (0.01)        & (0.01)        & (0.01)        \\
$\log$ frequency            & $0.45^{**}$   & $0.52^{**}$   & $0.42^{**}$   & $0.37^{**}$   & $0.34^{**}$   & $0.29^{**}$   \\
                            & (0.004)       & (0.02)        & (0.003)       & (0.01)        & (0.005)       & (0.01)        \\
Ticket fare (\$100)         & $-0.60^{**}$  & $-0.12^{**}$  & $-0.46^{**}$  & $-0.03^{**}$  & $-0.18^{**}$  & $-0.01$       \\
                            & (0.03)        & (0.02)        & (0.03)        & (0.01)        & (0.01)        & (0.004)       \\
$\rho$                      & $0.36^{**}$   & $0.22^{**}$   & $0.51^{**}$   & $0.31^{**}$   & $0.36^{**}$   & $0.24^{**}$   \\
                            & (0.005)       & (0.01)        & (0.004)       & (0.004)       & (0.01)        & (0.004)       \\
\midrule
Observations                & \num{89081}   & \num{27869}   & \num{183727}  & \num{42905}   & \num{82905}   & \num{36682}   \\
Adjusted R$^{2}$            & 0.82          & 0.77          & 0.72          & 0.76          & 0.55          & 0.64          \\
\bottomrule
\end{tabular}
\begin{tablenotes}[flushleft]\footnotesize
\item \emph{Notes:} Standard errors in parentheses. $^{**}$ denotes significance at 1\% level.
\end{tablenotes}
\end{threeparttable}

%% file: baseline_sens_30.tex
\robustify\bfseries
\sisetup{
  detect-weight=true,
  detect-family=true,
  retain-explicit-plus
}
\begin{threeparttable}
\setlength{\tabcolsep}{5pt}
\small
\begin{tabular}{
  l
  l
  S[table-format=+3.0]
  S[table-format=+3.2]
  S[table-format=+3.2]
  S[table-format=+3.2]
  S[table-format=+3.2]
  S[table-format=+3.2]
}
\toprule
 & & {\textbf{Charges}} & \multicolumn{5}{c}{\textbf{Welfare outcomes (\,\euro~B)}} \\
\cmidrule(lr){3-3}\cmidrule(lr){4-8}
\textbf{Scenario} & \textbf{Region} & {(\euro/ton)} & {Welfare} & {Emissions} & {Producers} & {Consumers} & {Government} \\
\midrule

\multirow{2}{*}{Baseline}
  & Europe   & 0  & 187.17 & -13.88 & 56.63 & 144.72 & 0.00 \\
  & America  & 0  & 369.42 & -29.27 & 151.01 & 257.97 & 0.00 \\
\cmidrule(l){2-8}
  & \textbf{Total ($\mathcal{W}$)} &  & \bfseries 556.59 & \bfseries -43.15 & \bfseries 207.64 & \bfseries 402.69 & \bfseries 0.00 \\
\bottomrule
\end{tabular}
\end{threeparttable}

%% file: sensitivity_30.tex
\robustify\bfseries
\sisetup{
  detect-weight=true,
  detect-family=true,
  retain-explicit-plus
}
\begin{threeparttable}
\setlength{\tabcolsep}{5pt}
\small
\begin{tabular}{
  l
  l
  S[table-format=+3.0] 
  S[table-format=+3.2] 
  S[table-format=+3.2] 
  S[table-format=+3.2] 
  S[table-format=+3.2] 
  S[table-format=+3.2] 
}
\toprule
 & & {\textbf{Charges}} & \multicolumn{5}{c}{\textbf{Welfare variations (\,\euro~B)}} \\
\cmidrule(lr){3-3}\cmidrule(lr){4-8}
\textbf{Scenario} & \textbf{Region} & {(\euro/ton)} & {Welfare} & {Emissions} & {Prod.} & {Cons.} & {Gov.} \\
\midrule

\multirow{2}{*}{Baseline}
  & Europe   & 0  & 0 & 0 & 0 & 0 & 0 \\
  & America  & 0  & 0 & 0 & 0 & 0 & 0 \\
  \cmidrule(l){2-8}
  & \textbf{Total ($\mathcal{W}$)} &  & \bfseries 0 & \bfseries 0$^{*}$ & \bfseries 0 & \bfseries 0 & \bfseries 0$^{*}$ \\
\addlinespace[1.5ex]

\multirow{2}{*}{\shortstack[l]{Local\\regulators}}
  & Europe   & -4  & -5.36 & -1.87 & +2.53 & +2.61 & -0.94 \\
  & America  & -68 & +13.41 & -1.87 & +28.64 & +16.18 & -37.23 \\
  \cmidrule(l){2-8}
  & \textbf{Total ($\mathcal{W}$)} &  & \bfseries +8.05 & \bfseries -3.74 & \bfseries +31.17 & \bfseries +18.79 & \bfseries -38.32 \\
\addlinespace[1.5ex]

\multirow{2}{*}{\shortstack[l]{Global\\regulator}}
  & Europe   &  & -13.28 & -3.29 & +16.21 & +7.36 & -33.56 \\
  & America  &  & +23.69 & -3.29 & +39.86 & +20.68 & -33.56 \\
  & Global   & -81  &   &   &   &   &  \\
  \cmidrule(l){2-8}
  & \textbf{Total ($\mathcal{W}$)} &  & \bfseries +10.42 & \bfseries -6.58 & \bfseries +56.07$^{*}$ & \bfseries +28.04$^{*}$ & \bfseries -67.12 \\
\addlinespace[1.5ex]

\multirow{2}{*}{\shortstack[l]{Global regulator\\discriminating}}
  & Europe   &  & -13.81 & -3.21 & +14.57 & +7.08 & -32.25 \\
  & America  &  & +24.23 & -3.21 & +39.17 & +20.52 & -32.25 \\
  & Global   & \text{EU: -72 \& NA: -81} &   &   &   &   &  \\
  \cmidrule(l){2-8}
  & \textbf{Total ($\mathcal{W}$)} &  & \bfseries +10.44$^{*}$ & \bfseries -6.42 & \bfseries +53.74 & \bfseries +27.61 & \bfseries -64.50 \\
\addlinespace[1.5ex]

\multirow{3}{*}{Overlapping}
  & Europe   & 66  & -10.67 & -2.59 & +5.10 & +4.08 & -17.26 \\
  & America  & -4  & +19.70 & -2.59 & +38.25 & +19.71 & -35.67 \\
  & Global\tnote{$\dagger$} & -83 &  &  &  &  & 0 \\
  \cmidrule(l){2-8}
  & \textbf{Total ($\mathcal{W}$)} &  & \bfseries +9.04 & \bfseries -5.18 & \bfseries +43.35 & \bfseries +23.79 & \bfseries -52.93 \\

\bottomrule
\end{tabular}
\begin{tablenotes}[flushleft]\footnotesize
\item[$^{*}$] Associated with the optimal scenario for each welfare component.
\item[$\dagger$] The global rows represent a distinct regulator redistributing half of the tax revenues to each of the two regions.
\end{tablenotes}
\end{threeparttable}